\documentclass[12pt,a4paper]{amsart}
\usepackage[T1]{fontenc}
\usepackage[utf8]{inputenc}
\usepackage[english]{babel}
\usepackage[all]{xy}
\usepackage[vmargin=72pt,hmargin=60pt]{geometry}
\newtheorem{theorem}{Theorem}[section]

\newtheorem{lemma}[theorem]{Lemma}
\newtheorem{proposition}[theorem]{Proposition}
\theoremstyle{definition}
\newtheorem{definition}[theorem]{Definition}
\theoremstyle{remark}
\newtheorem{example}[theorem]{Example}
\newtheorem{remark}[theorem]{Remark}

\numberwithin{equation}{section}
\newcommand{\norm}[1]{\left\Vert#1\right\Vert}

\newcommand{\prth}[1]{\left(#1\right)}
\newcommand{\lie}[1]{\left[#1\right]}

\newcommand{\set}[1]{\left\{#1\right\}}
\newcommand{\rbra}{\left(} \newcommand{\rket}{\right)}
\newcommand{\sbra}{\left[} \newcommand{\sket}{\right]}
\newcommand{\cbra}{\left\{} 
\newcommand{\abra}{\left\langle} \newcommand{\aket}{\right\rangle}
\newcommand{\nbra}{\left.} \newcommand{\nket}{\right.}
\newcommand{\Cartan}{\mathcal C}
\newcommand{\cns}{\mathcal C}
\newcommand{\CC}{\mathcal C}
\newcommand{\DD}{\mathcal D}

\newcommand{\smooth}{\CC^\infty}

\newcommand{\FF}{\mathcal F}
\newcommand{\hh}{\mathbf h}
\newcommand{\HH}{\mathcal H}
\newcommand{\LL}{\mathcal L}
\newcommand{\KK}{\mathcal K}
\newcommand{\MM}{\mathcal M}
\newcommand{\NN}{\mathcal N}
\newcommand{\nat}{\mathbb N}
\newcommand{\Lact}{\mathop{\mathcal S}}
\newcommand{\RR}{\mathbb R}

\newcommand{\eps}{\varepsilon}

\newcommand{\To}{\longrightarrow}
\newcommand{\tensor}{\otimes}
\newcommand{\diff}{\,\mathrm d}
\newcommand{\dd}[2][{}]{\frac{\!\diff#1}{\!\diff#2}}
\newcommand{\dde}{\dd\eps}

\newcommand{\pp}[2][{}]{\frac{\partial#1}{\partial#2}}

\newcommand{\dx}{\diff x}
\newcommand{\dmx}{\diff^mx}
\newcommand{\dmxi}{\diff^{m-1}x_i}
\newcommand{\dmxj}{\diff^{m-1}x_j}
\newcommand{\dy}{\diff y}
\newcommand{\du}{\diff u}
\newcommand{\dt}{\diff t}
\newcommand{\Id}{\mathit{Id}}

\newcommand{\pr}{\mathop{\mathit{pr}}\nolimits}
\newcommand{\Lie}{\mathfrak L}
\newcommand{\Fields}{\mathfrak X}
\newcommand{\Forms}{\mathit{\Omega}}
\newcommand{\Sections}{\Gamma}
\newcommand{\mathrb}[1]{\mathrm{\mathbf{#1}}}
\newcommand{\vv}{\mathrb v}
\newcommand{\ff}{\mathrb f}
\newcommand{\cf}{\emph{cf.}}
\newcommand{\eg}{\emph{e.g.}}
\newcommand{\etal}{\emph{et~al.}}
\newcommand{\ie}{\emph{i.e.}}

\newcommand{\secref}[1]{\S\ref{#1}}

\begin{document}
\title{Constrained Variational Calculus for Higher Order Classical Field Theories}
\author[C. M. Campos]{Cédric M. Campos}
\address{Instituto de Ciencias Matemáticas\\CSIC-UAM-UC3M-UCM\\Serrano 123\\28006 Madrid\\Madrid (Spain)}
 \email[C. M. Campos]{cedricmc@icmat.es}
\author[M. de León]{Manuel de León}
 \email[M. de León]{mdeleon@icmat.es}
\author[D. Martín de Diego]{David Martín de Diego}
 \email[D. Martín de Diego]{david.martin@icmat.es}
\date{May 12, 2010}
\begin{abstract}
We develop an intrinsic geometrical setting for  higher order constrained field theories. As a main tool we use an appropriate  generalization of the classical  Skinner-Rusk formalism. Some examples of application are studied, in particular, applications to the geometrical description of  optimal control theory for partial differential equations.
\end{abstract}
\maketitle

\section{Introduction} \label{sec:introduction}
 
Classical field theories can be intrinsically described in terms of the geometry of a fiber bundle $\pi:E\to M$ and its associated higher order jet bundles, $J^k\pi$ ($J^1\pi$ for the case of first order field theories). The jet bundles give a geometrical description for higher order partial derivatives of the fiber coordinates of $E$ with respect to those of $M$.

 As in almost every physical theory, the dynamics of a classical field theory is completely determined by a Lagrangian function, that is, a function on the corresponding jet bundle, $L: J^k\pi\to\RR$. For a theory of order $k$, the solutions are sections $\phi$ of the fiber bundle $\pi:E\to M$ such that they extremize the functional
\[ \Lact(\phi) = \int_RL(j^k\phi)\eta, \]
where $\eta$ is a fixed volume form (it is assumed that $M$ is orientable), $R\subseteq M$ is a compact set and $j^k\phi$ is the $k$-jet prolongation of $\phi$.

The most basic result on variational calculus is the construction from the above functional of a set of partial differential equations, the Euler-Lagrange equations, which must be satisfied by any smooth extremal. More interesting, the property of extremizing the problem does not depend on the particular chosen coordinate system (fact noted by J. L. Lagrange during his studies of analytical mechanics), therefore it must be able to write the Euler-Lagrange equations in an intrinsic way.

In this sense and restricting ourselves to first order field theories, when $L:J^1\pi\to M$, the fundamental object is the so-called Poincaré-Cartan form $\Omega_L$ which is an $(m+1)$-form ($\dim M=m$) univocally associated to the Lagrangian. This form is constructed using the geometry of the jet bundle and it is also related with the variational background \cite{GiMmsyI}. Using this form, it is possible to write down the Euler-Lagrange equations in an intrinsic way. Indeed, $\phi$ satisfies the Euler-Lagrange equations (that is, it is a critical point of the action $\Lact$) if and only if
\[ (j^1\phi)^*(i_V\Omega_L)=0, \quad\hbox{for all tangent vector $V$ in $TJ^1\pi$\;}. \]
Moreover, this form plays an important role in the connection between symmetries and conservation laws (see \cite{LnMrtSntm04}).

But many of the Lagrangians which appear in field theories are of higher order (as for instance in elasticity or gravitation), therefore it is interesting to find a fully geometric setting also for these field theories. Besides in first order field theories it is possible to define a Cartan $(m+1)$-form, one of the main difficulties in higher order field theories is that the uniqueness of the Cartan $(m+1)$-form is not guaranteed (contrary to the first order case). In other words, there will be different Cartan forms which carry out the same function in order to define an intrinsic formulation of Euler-Lagrange equations. The main reason of this problem, is the commutativity of repeated partial differentiation. In the literature, we find different approaches to fix the Cartan form for higher order field theories. For instance the approach by Arens \cite{Arns81} which consists of injecting the higher-order Lagrangian to a first-order one in an appropriate space by the introduction into the theory of a great number of variables and Lagrange multipliers. Following this perspective, and from  a more geometrical point of view, the reader is refereed to the papers by Aldaya and Azcárraga \cite{AldAzc78a,AldAzc80a}. A different approach is adopted by García and Muñoz whom described a method of constructing global Poincaré-Cartan forms in the higher order calculus of variations in fibered spaces by means of linear connections (see \cite{GrcMnz82,GrcMnz90}), in particular they show that the Cartan forms depend on the choice of two connections, a linear connection on the base $M$ and a linear connection on the vertical bundle $\mathcal{V}\pi$. In 1990, Crampin and Saunders \cite{SndCrmp90} proposed the use of an operator analogous to the almost tangent structure canonically defined on the tangent bundle of a given configuration manifold $M$ for the construction of global Poincaré-Cartan forms. In \cite{CmpsLnMrt09} we have found an intrinsic formalism for higher order field theories avoiding the problem of the degree of arbitrariness in the definition of Cartan forms for a given lagrangian function. We have proposed a differential geometric setting for the equations of motion derived from a higher order Lagrangian function, strongly based on the so-called Skinner and Rusk formalism for mechanics (see also the paper by Vitagliano \cite{Vtgl09}).

In this paper, we introduce constraints in the picture. The constraints are geometrically defined as a submanifold $\cns$ of $J^k\pi$ and a function $L:J^k\pi\to\RR$. In other words, we impose the constraints on the space of sections where the action is defined. We will show that the formalism introduced in our previous paper \cite{CmpsLnMrt09} is adapted to the case of constrained field theories, deriving an intrinsic framework of the constrained Euler-Lagrange equations. For the geometrical description, we start with the fibration $\pi_{W,M}:W\to M$, where $W$ is the fibered product $J^k\pi\times_{J^{k-1}\pi}\Lambda^m_2(J^{k-1}\pi)$, and we induce a submanifold $W^\cns_0$ of $W$ using the constraints given by $\cns$. Using that $W$ has a naturally (pre-)multisymplectic form (by pulling back the canonical multisymplectic form on $\Lambda^m_2(J^{k-1}\pi)$), we deduce an intrinsic and unique expression for a Cartan type equation for the Euler-Lagrange equations for constrained higher-order field theories. Additionally, we obtain a resultant constraint algorithm and we give conditions to ensure that the final constraint submanifold is multisymplectic. Finally, we discuss some examples to illustrate the theory.

\section{Notation} \label{sec:notation}
Throughout the paper, we will stick to the following notation. It will be reminded when necessary but, for the ease of the reader, we put it here together.

Lower case Latin (resp. Greek) letters will usually denote indexes that range between $1$ and $m$ (resp. $1$ and $n$). Capital Latin letters will usually denote multi-indexes whose length ranges between $0$ and $k$. In particular and if nothing else it is stated, $I$ and $J$ will usually denote multi-indexes whose length goes from $0$ to $k-1$ and $0$ to $k$, respectively; and $K$ (and sometimes $R$) will denote multi-indexes whose length is equal to $k$. The Einstein notation for repeated indexes and multi-indexes is understood but, for clarity, in some cases the summation for multi-indexes will be indicated.

We denote by $(E,\pi,M)$ a fiber bundle whose base space $M$ is a smooth manifold of dimension $m$, and whose fibers have dimension $n$, thus $E$ is $(m+n)$-dimensional. Adapted coordinate systems on $E$ will be of the form $(x^i,u^\alpha)$, where $(x^i)$ is a local coordinate system on $M$ and $(u^\alpha)$ denotes fiber coordinates. The local volume form associated to $(x^i)$ is $\dmx=\dx^1\wedge\dots\wedge\dx^m$. We write $\dmxi$ for the contraction $i_{\partial/\partial x^i}\dmx$. Note that, for a fixed index $i\in\set{1,\dots,m}$, we have that $\dx^i\wedge\dmxi=\dmx$.

In section \secref{sec:variational.calculus} and all later sections, $M$ shall be supposed oriented and a volume form $\eta$ will be provided. Local coordinates $(x^i)$ on $M$ will then be considered subject to compatibility with $\eta$, that is, such that $\dmx=\eta$. Any $m$-form along a fiber projection over $M$ will be denoted with a curly letter. The same letter in a straight font shall denote its factor with respect to either, the volume form $\eta$ or the local one $\dmx$, \eg\ $\LL=L\dmx$, $\HH=H\dmx$, etc.

Finally, pullbacks of forms and functions by fiber projections will still be denoted by the same symbol. For instance, $\pi^*\eta$ will still be denoted $\eta$ on $E$.

\section{Jet Bundles} \label{sec:jet.bundles}

We fix a fiber bundle $\pi:E\to M$ with $\dim M=m$ and $\dim E=m+n$, and where neither $M$ nor $E$ are necessarily orientable.

Let $\phi,\psi:M\To E$ be two smooth local sections of $\pi$ around a given point $x\in M$. We say that $\phi$ and $\psi$ are \emph{$k$-equivalent} at $x$ (with $k\geq0$) if the sections and all their partial derivatives until order $k$ coincide at $x\in M$, that is, if
\[ \phi(x)=\psi(x)\textrm{ and } \left.\frac{\partial^k\phi^\alpha}{\partial x^{i_1}\cdots\partial x^{i_k}}\right|_x=\left.\frac{\partial^k\psi^\alpha}{\partial x^{i_1}\cdots\partial x^{i_k}}\right|_x, \]
for all $1\leq\alpha\leq n$, $1\leq i_j\leq m$, $1\leq j\leq k$. Note that this is independent of the chosen coordinate system (adapted or not) and, therefore, to be $k$-equivalent is an equivalence relation (see \cite{LnRdrs85,LnRdrs89,Snd89}, for more details).

\begin{definition} \label{def:jet.bundle}
Let $x\in M$, given a smooth local section $\phi\in\Sections_x\pi$, the equivalence class of $k$-equivalent smooth local sections (with $k\geq0$) around $x$ that contains $\phi$ is called the \emph{$k$th jet of $\phi$ at $x$} and is denoted $j^k_x\phi$. The \emph{$k$th jet manifold of $\pi$}, denoted $J^k\pi$, is the whole collection of $k$th jets of arbitrary local sections of $\pi$, that is,
\begin{equation} \label{eq:jet.bundle}
J^k\pi = \set{j^k_x\phi\ :\ x\in M,\ \phi\in\Sections_x\pi}.
\end{equation}
Given a (local) section of $\pi$, its \emph{$k$th lift} is $(j^k\phi)(x)=j^k_x\phi$.
\end{definition}

Adapted coordinates $(x^i,u^\alpha)$ on the total space $E$ induce coordinates $(x^i,u^\alpha_I)$ (with $0\leq|I|\leq k$) on the $k$-jet manifold $J^k\pi$ given by:
\[ u^\alpha_I(j^k_x\phi) = \left.\pp[^{|I|}\phi^\alpha]{x^I}\right|_x, \]
from where we deduce that $J^k\pi$ is actually a smooth manifold of dimension
\[ \dim J^k\pi = m + n\cdot\sum_{l=0}^k\binom{m-1+l}{m-1}. \]
It is readily seen that $(J^k\pi,\pi_k,M)$ is a fiber bundle, where
\[ \pi_k(j^k_x\phi) = x  \quad  (\textrm{in coordinates }\pi_k(x^i,u^\alpha_I) = (x^i)). \]
It is also clear that the $k$-jet manifold $J^k\pi$ fibers over the lower order $l$-jet manifolds $J^l\pi$ (see Diagram \ref{fig:jet.chain}), with $0\leq l<k$, where by definition $J^0\pi=E$ and where the projections are given by:
\[ \pi_{k,l}(j^k_x\phi) = j^l_x\phi  \quad  (\textrm{in coordinates }\pi_{k,l}(x^i,u^\alpha_I) = (x^i,u^\alpha_J), \textrm{ with }0\leq|I|\leq k,0\leq|J|\leq l). \]
\begin{figure}[!h]
\[\xymatrix{
  J^k\pi \ar[r]^{\pi_{k,k-1}} \ar[ddrrrr]_{\pi_k} & \cdots \ar[r]^{\pi_{3,2}} & J^2\pi \ar[r]^{\pi_{2,1}} \ar[ddrr]_{\pi_2} & J^1\pi \ar[r]^{\pi_{1,0}} \ar[ddr]_{\pi_1} & E \ar[dd]^\pi \\ \\
  & & & & M \ar@/_1.2pc/[uu]_\phi \ar@/^2pc/[uullll]^{j^k\phi}
  }\]
\caption{Chain of jets} \label{fig:jet.chain}
\end{figure}
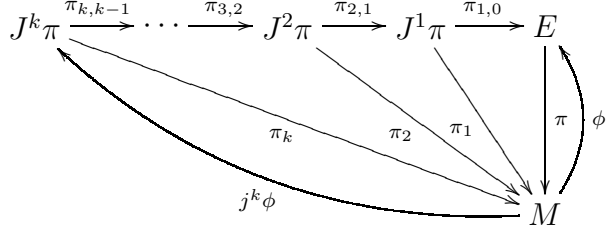

\begin{definition} \label{def:dual.jet.bundle}
The \emph{$k$th dual jet bundle}, denoted $J^k\pi^\dag$, is the subbundle of $\Lambda^mJ^{k-1}\pi$ of $\pi_{k-1}$-semi-basic $m$-forms over $J^{k-1}\pi$, that is,
\begin{equation} \label{eq:dual.jet.bundle}
J^k\pi^\dag = \Lambda^m_2J^{k-1}\pi =  \set{ \omega\in\Lambda^mJ^{k-1}\pi\ :\ i_{v_1}i_{v_2}\omega=0,\ v_1,v_2\in\mathcal{V}\pi_{k-1} },
\end{equation}
where $\mathcal{V}\pi_{k-1}$ is the vertical bundle of $\pi_{k-1}:J^{k-1}\pi\to M$.
\end{definition}

\begin{remark}
The original definition of the $k$th dual jet bundle is as the extended dual affine bundle of the iterated jet bundle $(J^1\pi_{k-1},(\pi_{k-1})_{1,0},J^{k-1}\pi)$. Observe that $J^k\pi$ is naturally embedded as an affine subbundle of $J^1\pi_{k-1}$.
\end{remark}

Locally, the elements of $J^k\pi^\dag$ are of the form
\[ p\dmx + p^{Ii}_\alpha\du^\alpha_I\wedge\dmxi, \]
where $0\leq|I|\leq k-1$. Thus, adapted coordinates $(x^i,u^\alpha)$ on $E$ induce coordinates $(x^i,u^\alpha_I,p,p^{Ii}_\alpha)$ on $J^k\pi^\dag$. Note that $(x^i,u^\alpha_I)$ are adapted coordinates on $J^{k-1}\pi$ and that $(p,p^{Ii}_\alpha)$ are fiber coordinates.

While the $k$th jet bundle $J^k\pi$ projects over the lower order jet bundles (Diagram \ref{fig:jet.chain}), the $k$th dual jet bundle is ``embedded'' into the $(k+1)$th dual jet bundle by means of the pullback of the affine projection $\pi_{k+1,k}$ (Diagram \ref{fig:dual.jet.chain}).

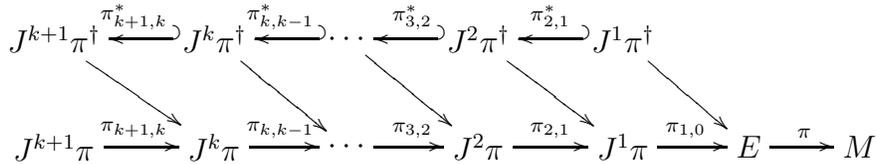
\begin{figure}[!h]
\[\xymatrix{
  J^{k+1}\pi^\dag \ar[dr] & J^k\pi^\dag \ar[dr] \ar@{_{(}->}[l]_{\pi_{k+1,k}^*} & \cdots \ar[dr] \ar@{_{(}->}[l]_{\pi_{k,k-1}^*} & J^2\pi^\dag \ar[dr] \ar@{_{(}->}[l]_{\pi_{3,2}^*} & J^1\pi^\dag \ar[dr] \ar@{_{(}->}[l]_{\pi_{2,1}^*} \\
  J^{k+1}\pi \ar[r]^{\pi_{k+1,k}} & J^k\pi \ar[r]^{\pi_{k,k-1}} & \cdots \ar[r]^{\pi_{3,2}} & J^2\pi \ar[r]^{\pi_{2,1}} & J^1\pi \ar[r]^{\pi_{1,0}} & E \ar[r]^\pi & M
  }\]
\caption{Chain of dual jets} \label{fig:dual.jet.chain}
\end{figure}

The dual jet bundle $J^k\pi^\dag$ has a canonical multi-symplectic structure (see \cite{CntrIbrtLn96,CntrIbrtLn99,CrnyCrmpIbrt91}) and its elements are naturally paired with those of $J^k\pi$. We recall that a form $\Omega$ is multi-symplectic if it is closed and if its contraction with a single tangent vector is injective, that is, $i_V\Omega=0$ if and only if $V=0$.

\begin{definition} \label{def:tautological.form}
The \emph{tautological $m$-form} on $J^k\pi^\dag$ is the form given by
\begin{equation} \label{eq:tautological.form}
\Theta_\omega(V_1,\dots,V_m) = (\tau_{J^k\pi^\dag}^*\omega)(V_1,\dots,V_m),\ \omega\in J^k\pi^\dag,\ V_1,\dots,V_m\in T_\omega J^k\pi^\dag,
\end{equation}
where $\tau_{J^k\pi^\dag}$ is the natural projection from $TJ^k\pi^\dag$ to $J^k\pi^\dag$.
The \emph{canonical multi-symplectic $(m+1)$-form} on $J^k\pi^\dag$ is
\begin{equation} \label{eq:canonical.form}
\Omega = -\diff\Theta.
\end{equation}
\end{definition}

\begin{definition} \label{def:natural.pairing}
The natural pairing between $J^k\pi$ and its dual $J^k\pi^\dag$ is the fibered map $\Phi:J^k\pi\times_{J^{k-1}\pi}J^k\pi^\dag\to\Lambda^mM$ given by
\begin{equation} \label{eq:natural.pairing}
\Phi(j^k_x\phi,\omega)=(j^{k-1}\phi)^*_{j^{k-1}_x\phi}\omega.
\end{equation}
\end{definition}

Let $(x^i,u^\alpha_I,u^\alpha_K)$ and $(x^i,u^\alpha_I,p,p^{Ii}_\alpha)$, where $|I|=0,\dots,k-1$ and $|K|=k$, denote adapted coordinates on $J^k\pi$ and $J^k\pi^\dag$, respectively. Then, the tautological form and the canonical one are locally written
\begin{equation} \label{eq:canonical.form.coord}
\Theta = p\dmx+p^{Ii}_\alpha\du^\alpha_I\wedge\dmxi
\quad\textrm{and}\quad
\Omega = -\diff p\wedge\dmx-\diff p^{Ii}_\alpha\wedge\du^\alpha_I\wedge\dmxi,
\end{equation}
and the fibered pairing between the elements of $J^k\pi$ and $J^k\pi^\dag$ is locally written
\begin{equation} \label{eq:natural.pairing.coord}
\Phi(x^i,u^\alpha_I,u^\alpha_K,p,p^{Ii}_\alpha) = (p+p^{Ii}_\alpha u^\alpha_{I+1_i})\dmx.
\end{equation}

With the aim of clarifying the latter sections and to make this work self content, we continue presenting some classical definitions besides of some basic properties, which may be found in \cite{KrpkSnd08,Olvr86,Snd89}. 

\begin{definition} \label{def:kth.prolongation.morphism}
Let $f:E\To F$ be a morphism between two fiber bundles $(E,\pi,M)$ and $(F,\rho,N)$, such that the induced function on the base, $\check f:M\To N$, is a diffeomorphism. The \emph{$k$th prolongation of $f$} is the map $j^kf:J^k\pi\To J^k\rho$ given by
\begin{equation} \label{eq:kth.prolongation.morphism}
(j^kf)(j^k_x\phi) := j^k_{\check f(x)}\phi_f,\ \forall j^k_x\phi\in J^k\pi,
\end{equation}
where $\phi_f:=f\circ\phi\circ\check f^{-1}$.
\end{definition}

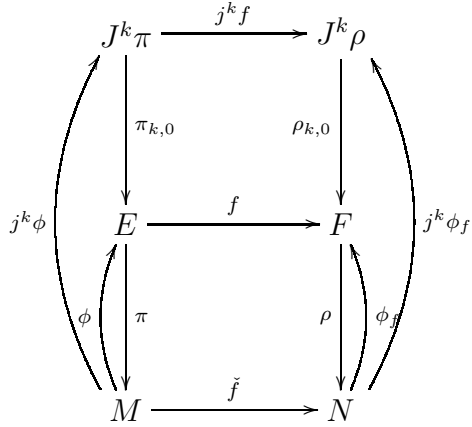
\begin{figure}[!h]
\[ \xymatrix{
  J^k\pi \ar[rr]^{j^kf}     \ar[dd]^{\pi_{k,0}} && J^k\rho \ar[dd]_{\rho_{k,0}} \\ \\
       E \ar[rr]^f          \ar[dd]^\pi         &&       F \ar[dd]_\rho         \\ \\
       M \ar[rr]^{\check f} \ar@/^0.8pc/[uu]^\phi     \ar@<1ex>@/^1.8pc/[uuuu]^{j^k\phi} &&
       N                    \ar@/_0.8pc/[uu]_{\phi_f} \ar@<-1.2ex>@/_1.8pc/[uuuu]_{j^k\phi_f}
} \]
\caption{The $k$th prolongation of a morphism} \label{fig:kth.prolongation.morphism}
\end{figure}

Note that the $k$th prolongation $j^kf$ of a morphism $f$ is both, a morphism between $(J^k\pi,\pi_{k,0},E)$ and $(J^k\rho,\rho_{k,0},F)$, and a morphism between $(J^k\pi,\pi_k,M)$ and $(J^k\rho,\rho_k,N)$. In each case, the induced functions between the base spaces are $f$ and $\check f$, respectively.

\begin{definition} \label{def:total.derivative}
The \emph{(coordinate) total derivatives} are the vector fields along $\pi_{k,k-1}$ locally given by
\begin{equation} \label{eq:total.derivative}
\dd{x^i} = \pp{x^i} + u^\alpha_{I+1_i}\pp{u^\alpha_I} = \pp{x^i} + u^\alpha_i\pp{u^\alpha} + u^\alpha_{ji}\pp{u^\alpha_j} + \dots
\end{equation}
for coordinates $(x^i,u^\alpha_J)$ on $J^k\pi$.
\end{definition}

Let $f\in\smooth(J^k\pi)$, then $\dd[f]{x^i}\in\smooth(J^{k+1}\pi)$ on the corresponding local chart. In general and as the coordinate total derivatives commute, $\dd[^{|J|}f]{x^J}\in\smooth(J^{k+|J|}\pi)$ for any multi-index $J\in\nat^m$. Observe also that total derivatives are closely related to partial derivatives by the relation
\begin{equation} \label{eq:total.and.partial.derivatives}
\dd[^{|J|}f]{x^J}\circ j^{k+|J|}\phi = \pp[^{|J|}(f\circ j^k\phi)]{x^J}
\end{equation}
for any section $\phi\in\Sections\pi$.

\begin{definition} \label{def:contact.form}
A \emph{contact 1-form} on $J^k\pi$ is a 1-form $\theta\in\Lambda^1J^k\pi$ which is pulled back to the zero form on $M$ by the $k$th-lift of any section $\phi$ of $\pi$, that is,
\begin{equation} \label{eq:contact.form}
(j^k\phi)^*\theta=0,\ \forall\phi\in\Sections\pi.
\end{equation}
The set $\Cartan^k$ of all the contact forms is called the \emph{Cartan distribution (of order $k$)}.
\end{definition}

\begin{proposition} \label{th:contact.form.coordinates}
Let $(x^i,u^\alpha_J)$ be adapted coordinates on $J^k\pi$, a basis of the Cartan distribution is given by the contact forms
\begin{equation} \label{eq:contact.form.coordinates}
\theta^\alpha_I = \du^\alpha_I - u^\alpha_{I+1_i}\dx^i,\ 0\leq|I|\leq k-1.
\end{equation}
\end{proposition}


\begin{definition} \label{def:klift.vector}
Given a vector field $\xi$ on $E$, its \emph{$k$th-lift} (or \emph{$k$th-jet}) is the unique vector field $\xi^{(k)}$ on $J^k\pi$ that is projectable to $\xi$ by $\pi_{k,0}$ and preserves the Cartan distribution with respect to the Lie derivative.
\end{definition}

\begin{proposition} \label{th:klift.vector}
Let $\xi$ be a vector field on $E$. If $\xi$ has the local expression
\begin{equation}
\xi = \xi^i\pp{x^i} + \xi^\alpha\pp{u^\alpha}
\end{equation}
in adapted coordinates $(x^i,u^\alpha)$ on $E$, then its $k$th-lift $\xi^{(k)}$ has the form
\begin{equation} \label{eq:klift.vector}
\xi^{(k)} = \xi^i\pp{x^i} + \xi^\alpha_J\pp{u^\alpha_J}
\end{equation}
for the induced coordinates $(x^i,u^\alpha_J)$ on $J^k\pi$, where
\begin{equation} \label{eq:klift.vector.components}
\xi^\alpha_0=\xi^\alpha \quad \textrm{and} \quad \xi^\alpha_{I+1_i}=\dd[\xi^\alpha_I]{x^i}-u^\alpha_{I+1_j}\dd[\xi^j]{x^i}.
\end{equation}
In particular, if $\xi$ is vertical with respect to $\pi$, then $\xi^\alpha_J=\diff^{|J|}\xi^\alpha/\dx^J$.
\end{proposition}

\begin{proposition} \label{th:klift.flow}
Let $\psi_\eps$ be the flow of a given $\pi$-projectable vector field $\xi$ over $E$. Then, the flow of $\xi^{(k)}$ is the $k$th-lift of $\psi_\eps$, $j^k\psi_\eps$.
\end{proposition}


\section{Variational Calculus} \label{sec:variational.calculus}

The dynamics in classical field theory is specified giving a Lagrangian density: A Lagrangian density is a mapping $\LL:J^k\pi\to\Lambda^mM$. Fixed a volume form $\eta$ on $M$, there is a smooth function $L:J^k\pi\to\RR$ such that $\LL=L\eta$.

\begin{definition} \label{def:integral.action}
Given a Lagrangian density $\LL:J^k\pi\To\Lambda^mM$, the associated \emph{integral action} is the map $\Lact:\Sections\pi\times\KK\To\RR$ given by
\begin{equation} \label{eq:integral.action}
\Lact(\phi,R)=\int_R(j^k\phi)^*\LL,
\end{equation}
where $\KK$ is the collection of smooth compact regions of $M$.
\end{definition}

\begin{definition} \label{def:variation}
Let $\phi$ be a section of $\pi$. A \emph{(vertical) variation} of $\phi$ is a curve $\eps\in I\mapsto\phi_\eps\in\Sections\pi$ (for some interval $I\subset\RR$) such that $\phi_\eps=\varphi_\eps\circ\phi\circ\check\varphi_\eps^{-1}$, where $\varphi_\eps$ is the flow of a (vertical) $\pi$-projectable vector field $\xi$ on $E$.
\end{definition}

When $\xi$ is vertical, then its flow $\varphi_\eps$ is an automorphism of fiber bundles over the identity for each $\eps\in I$, \ie\ $\check\varphi_\eps=\Id_M$.

\begin{definition} \label{def:critical.point}
We say that $\phi\in\Sections\pi$ is a \emph{critical} or \emph{stationary point} of the Lagrangian action $\Lact$ if and only if
\begin{equation} \label{eq:critical.point}
\dde\lie{\Lact(\phi_\eps,R)}\Big|_{\eps=0} = \dde\lie{\int_R(j^k\phi_\eps)^*\LL}\bigg|_{\eps=0} = 0,
\end{equation}
for any vertical variation $\phi_\eps$ of $\phi$ whose associated vertical field vanishes outside of $\pi^{-1}(R)$ and where $R_\eps=\check\varphi_\eps(R)$.
\end{definition}

\begin{lemma} \label{th:lie.derivative}
Let $\phi_\eps=\varphi_\eps\circ\phi\circ\check\varphi_\eps^{-1}$ be a variation of a section $\phi\in\Sections\pi$. If $\xi$ denotes the infinitesimal generator of $\varphi_\eps$, then
\begin{equation} \label{eq:lie.derivative}
\dde\lie{(j^k(\varphi_\eps\circ\phi)^*_x\omega}\Big|_{\eps=0} = (j^k\phi)^*_x(\Lie_{\xi^{(k)}}\omega),
\end{equation}
for any differential form $\omega\in\Forms(J^k\pi)$.
\end{lemma}

\begin{proof}
From Proposition \ref{th:klift.flow}, we have that $\xi^{(k)}$ is the infinitesimal generator of $j^k\varphi_\eps$. We then obtain by a direct computation,
\[ (j^k\phi)^*_x(\Lie_{\xi^{(k)}}\omega) = (j^k\phi)^*_x\prth{\dde\lie{(j^k\varphi_\eps)^*\omega}\Big|_{\eps=0}} = \dde\lie{(j^k\varphi_\eps\circ j^k\phi)^*_x\omega}\Big|_{\eps=0}. \]
\end{proof}

The following lemma will show to be useful in the variational derivation of the higher-order Euler-Lagrange equations.

\begin{lemma}[Higher-order integration by parts] \label{th:integration.by.parts}
Let $R\subset M$ be a smooth compact region and let $f,g:R\To\RR$ be two smooth functions. Given any multi-index $J\in\nat^m$, we have that
\begin{equation} \label{eq:integration.by.parts}
\int_R\pp[^{|J|}f]{x^J}g\dmx = (-1)^{|J|}\int_Rf\pp[^{|J|}g]{x^J}\dmx
+ \!\!\!\!\! \sum_{I_f+I_g+1_i=J} \!\!\!\!\! \lambda(I_f,I_g,J)\int_{\partial R}\pp[^{|I_f|}f]{x^{I_f}}\pp[^{|I_g|}g]{x^{I_g}}\dmxi,
\end{equation}
where $\lambda$ is given by the expression
\begin{equation} \label{eq:integration.by.parts.lambda}
\lambda(I_f,I_g,J) := (-1)^{|I_g|} \cdot \frac{|I_f|!\cdot|I_g|!}{|J|!} \cdot \frac{J!}{I_f!\cdot I_g!}.
\end{equation}
\end{lemma}

\begin{proof}
In this proof, we will use the shorthand notation $f_J=\pp[^{|J|}f]{x^J}$ so as to safe space.

We proceed by induction on the length $l$ of the multi-index $J$. The case  $l=|J|=0$ is a trivial identity and the case $l=|J|=1$ is the well known formula of integration by parts
\[ \int_Rf_ig\dmx =  \int_{\partial R}fg\dmxi - \int_Rfg_i\dmx. \]
Thus, let us suppose that the result is true for any multi-index $J\in\nat^m$ up to length $l>1$, in order to show that it is also true for any multi-index $K\in\nat^m$ of length $l+1$. Let $J$ and $1\leq j\leq m$ such that $J+1_j=K$. We then have,
\begin{eqnarray*}
\int_Rf_{J+1_j}g\dmx
&=& -\int_Rf_Jg_j\dmx + \int_{\partial R}f_Jg\dmxj\\
&=& (-1)^{l+1}\int_Rfg_{J+1_j}\dmx + \int_{\partial R}f_Jg\dmxj\\
& & - \!\!\!\!\! \sum_{I_f+I_{g_j}+1_i=J} \!\!\!\!\! \lambda(I_f,I_{g_j},i)\int_{\partial R}f_{I_f}g_{I_{g_j}+1_j}\dmxi,
\end{eqnarray*}
where we have used the first-order integration formula in first place, to then apply the induction hypothesis.
We now multiply each member by $(J(j)+1)/(l+1)$ and sum over $J+1_j=K$. Using the multi-index identity \eqref{eq:identity}, we have
\begin{eqnarray*}
\int_Rf_Kg\dmx
&=& (-1)^{l+1}\int_Rfg_K\dmx + \sum_{J+1_j=K}\frac{J(j)+1}{l+1}\int_{\partial R}f_Jg\dmxj\\
& & - \sum_{J+1_j=K}\frac{J(j)+1}{l+1}\sum_{I_f+I_{g_j}+1_i=J} \!\!\!\!\! \lambda(I_f,I_{g_j},i)\int_{\partial R}f_{I_f}g_{I_{g_j}+1_j}\dmxi.
\end{eqnarray*}
It only remain to rearrange properly the last two terms to express them in the stated form. Clearly,
\begin{multline*}
\sum_{J+1_j=K}\frac{J(j)+1}{l+1}\int_{\partial R}f_Jg\dmxj =\\
= \sum_{\substack{I_f+I_g+1_i=K\\|I_g|=0}} \!\! (-1)^{|I_g|} \cdot \frac{|I_f|!\cdot|I_g|!}{|K|!} \cdot \frac{K!}{I_f!\cdot I_g!} \int_{\partial R}\pp[^{|I_f|}f]{x^{I_f}}\pp[^{|I_g|}g]{x^{I_g}}\dmxi.
\end{multline*}
The last term is a little bit more tricky,
\begin{multline*}
\sum_{J+1_j=K}\frac{J(j)+1}{l+1} \!\! \sum_{I_f+I_{g_j}+1_i=J} \!\!\! (-1)^{|I_{g_j}|+1} \cdot \frac{|I_f|!\cdot|I_{g_j}|!}{|J|!} \cdot \frac{J!}{I_f!\cdot I_{g_j}!} \int_{\partial R}f_{I_f}g_{I_{g_j}+1_j}\dmxi =\\
\begin{split}
&= \!\! \sum_{I_f+I_{g_j}+1_i+1_j=K} \!\!\!\!\! (-1)^{|I_{g_j}|+1} \cdot \frac{|I_f|!\cdot|I_{g_j}|!}{|K|!} \cdot \frac{K!}{I_f!\cdot I_{g_j}!} \int_{\partial R}f_{I_f}g_{I_{g_j}+1_j}\dmxi\\
&= \!\! \sum_{\substack{I_f+I_g+1_i=K\\|I_g|\geq1}} \!\! (-1)^{|I_g|} \!\! \sum_{I_{g_j}+1_j=I_g} \frac{I_g(j)}{|I_g|} \cdot \frac{|I_f|!\cdot|I_g|!}{|K|!} \cdot \frac{K!}{I_f!\cdot I_g!} \int_{\partial R}f_{I_f}g_{I_g}\dmxi\\
&= \!\! \sum_{\substack{I_f+I_g+1_i=K\\|I_g|\geq1}} \!\! (-1)^{|I_g|} \cdot \frac{|I_f|!\cdot|I_g|!}{|K|!} \cdot \frac{K!}{I_f!\cdot I_g!} \int_{\partial R}f_{I_f}g_{I_g}\dmxi
\end{split}
\end{multline*}
where we have used the identity \eqref{eq:identity} again. The result is now clear.
\end{proof}

\begin{theorem}[The higher-order Euler-Lagrange equations] \label{th:euler-lagrange}
Given a fiber section $\phi\in\Sections\pi$, let us consider an infinitesimal variation $\phi_\eps$ of it such that the support $R$ of the associated vertical vector field $\xi$ is contained in a coordinate chart $(x^i)$. We then have that the variation of the Lagrangian action $\Lact$ at $\phi$ is given by
\begin{equation} \label{eq:critical.point.coord}
\begin{array}{rcl}
\displaystyle \dde\Lact(\phi_\eps,R_\eps)\Big|_{\eps=0}
&=& \displaystyle \sum_{|J|=0}^k \sbra (-1)^{|J|}\int_R(j^{2k}\phi)^*\prth{\xi^\alpha\dd[^{|J|}]{x^J}\pp[L]{u^\alpha_J}}\dmx \phantom{\sum_{I_\xi}}\nket\\
& & \displaystyle \nbra + \!\!\!\!\! \sum_{I_\xi+I_L+1_i=J} \!\!\!\!\! \lambda(I_\xi,I_L,J)\int_{\partial R} (j^{2k}\phi)^*\prth{\xi^\alpha_{I_\xi}\dd[^{|I_L|}]{x^{I_L}}\pp[L]{u^\alpha_J}}\dmxi \sket.
\end{array}
\end{equation}
Moreover, $\phi$ is a critical point of the Lagrangian action $\Lact$ if and only if it satisfies the \emph{higher-order Euler-Lagrange equations}
\begin{equation} \label{eq:euler-lagrange}
(j^{2k}\phi)^*\prth{ \sum_{|J|=0}^k(-1)^{|J|}\dd[^{|J|}]{x^J}\pp[L]{u^\alpha_J} } = 0
\end{equation}
on the interior of $M$, plus the boundary conditions
\begin{equation} \label{eq:boundary.conditions}
\dd[^{|I|}]{x^I}\pp[L]{u^\alpha_J}=0,\ 0\leq|I|<|J|\leq k,
\end{equation}
on the boundary $\partial M$ of $M$.
\end{theorem}

\begin{proof}
Let us denote by $\xi$ the vertical field associated to the variation $\phi_\eps$. By Proposition \ref{th:lie.derivative}, Cartan's formula $\Lie=\diff\circ i+i\circ\diff$ and Proposition \ref{th:klift.vector}, we have that
\begin{eqnarray*}
\dde\Lact(\phi_\eps,R)\Big|_{\eps=0}
&=& \int_R(j^k\phi)^*(\Lie_{\xi^{(k)}}\LL)\\
&=& \int_R(j^k\phi)^*\diff(i_{\xi^{(k)}}\LL) + \int_R(j^k\phi)^*i_{\xi^{(k)}}\diff\LL\\
&=& \int_{\partial R}(j^k\phi)^*i_{\xi^{(k)}}\LL + \int_R(j^k\phi)^*(\xi^{(k)}(L)\dmx-\diff L\wedge i_{\xi^{(k)}}\dmx)\\
&=& \int_R(j^k\phi)^*\prth{\sum_{|J|=0}^k\dd[^{|J|}\xi^\alpha]{x^J}\pp[L]{u^\alpha_J}}\dmx
\end{eqnarray*}
If we now apply the higher-order integration by parts \eqref{eq:integration.by.parts} and we take into account that Equation \eqref{eq:total.and.partial.derivatives}, we obtain that
\begin{eqnarray*}
\dde\Lact(\phi_\eps,R)\Big|_{\eps=0}
&=& \sum_{|J|=0}^k \sbra (-1)^{|J|}\int_R(j^{2k}\phi)^*\prth{\xi^\alpha\dd[^{|J|}]{x^J}\pp[L]{u^\alpha_J}}\dmx  \phantom{\sum_{I_\xi}}\nket\\
& & \nbra + \!\!\!\!\! \sum_{I_\xi+I_L+1_i=J} \!\!\!\!\! \lambda(I_\xi,I_L,J)\int_{\partial R} (j^{2k}\phi)^*\prth{\dd[^{|I_\xi|}\xi^\alpha]{x^{I_\xi}}\dd[^{|I_L|}]{x^{I_L}}\pp[L]{u^\alpha_J}}\dmxi \sket,
\end{eqnarray*}
which is the first statement of our theorem.

If we now suppose that $R$ is contained in the interior of $M$, as $\xi$ is null outside of $R$, so it is $\xi^{(k)}$ outside of $R$ and, by smoothness, on its boundary $\partial R$. Thus, if $\phi$ is a critical point of $\Lact$, we then must have that
\[ \dde\Lact(\phi_\eps,R)\Big|_{\eps=0} = \int_R(j^{2k}\phi)^*\prth{\xi^\alpha\sum_{|J|=0}^k(-1)^{|J|}\dd[^{|J|}]{x^J}\pp[L]{u^\alpha_J}}\dmx = 0, \]
for any vertical field $\xi$ whose compact support is contained in $\pi^{-1}(R)$. We thus infer that $\phi$ shall satisfy the higher-order Euler-Lagrange equations \eqref{eq:euler-lagrange} on the interior of $M$.

Finally, if $R$ has common boundary with $M$, we then have that
\begin{equation}
\dde\Lact(\phi_\eps,R)\Big|_{\eps=0} = \sum_{|J|=0}^k\sum_{I_\xi+I_L+1_i=J} \!\!\!\!\! \lambda(I_\xi,I_L,J)\int_{\partial R\cap\partial M} (j^{2k}\phi)^*\prth{\xi^\alpha_{I_\xi}\dd[^{|I_L|}]{x^{I_L}}\pp[L]{u^\alpha_J}}\dmxi = 0.
\end{equation}
As this is true for any vertical field $\xi$ whose compact support is contained in $\pi^{-1}(R)$, we deduce the boundary conditions \eqref{eq:boundary.conditions}.
\end{proof}

\begin{remark} \label{rmk:null.boundary}
If in the Definition \ref{def:critical.point} of a critical point, one further requires to the variations to be null at the boundary $\partial M$ of $M$, in the sense that the associated vector field $\xi$ be zero over $\pi^{-1}(\partial M)$, then would no longer have the boundary condition \eqref{eq:boundary.conditions}. In such a case, the space of solutions would be broader and one could impose boundary conditions to them in order to obtain a particular one.

We also note that the theorem remains true if we choose vector variations, non necessarily vertical. Although the proof would be much more cumbersome.
\end{remark}

\section{The Skinner-Rusk Formalism} \label{sec:sr.formalism}

What follows in this section has already been published in \cite{CmpsLnMrt09}. Here we give a summary of the most relevant facts that are proven in there, introducing some new insights and points of view. The generalization of the Skinner-Rusk formalism to higher order classical field theories will take place in the fibered product
\begin{equation} \label{eq:mixed.space}
W=J^k\pi\times_{J^{k-1}\pi}J^k\pi^\dag.
\end{equation}
The first order case is covered in \cite{LnMrrMrt03,MrnMnzRmn03} or \cite{CntrCrtMrtz02} for time dependent mechanics; see also \cite{SknRsk83a,SknRsk83b} for the original treatment by Skinner and Rusk. The projection on the $i$-th factor will be denoted $\pr_i$ (with $i=1,2$) and the projection as fiber bundle over $J^{k-1}\pi$ will be $\pi_{W,J^{k-1}\pi}=\pi_{k,k-1}\circ\pr_1$ (see Diagram \ref{fig:dynamical.framework}). On $W$, adapted coordinates are of the form $(x^i,u^\alpha_I,u^\alpha_K,p,p^{Ii}_\alpha)$, where $|I|=0,\dots,k-1$ and $|K|=k$. Note that $(x^i,u^\alpha_I)$ are coordinates on $J^{k-1}\pi$ and that $(u^\alpha_K)$ and $(p,p^{Ii}_\alpha)$ are fiber coordinates on $J^k\pi\To J^{k-1}\pi$ and $J^k\pi^\dag\To J^{k-1}\pi$, respectively.

\begin{figure}[!h]
\[\xymatrix{
  && W \ar[lld]_{\pr_1} \ar[dd]^{\pi_{W,J^{k-1}\pi}} \ar[rrd]^{\pr_2}&&\\
  J^k\pi \ar[rrd]^{\pi_{k,k-1}} && && J^k\pi^\dag \ar[lld] \\
  && J^{k-1}\pi \ar[d]^{\pi_{k-1}}&&\\
  && M &&
  }\]
\caption{The Skinner-Rusk framework} \label{fig:dynamical.framework}
\end{figure}
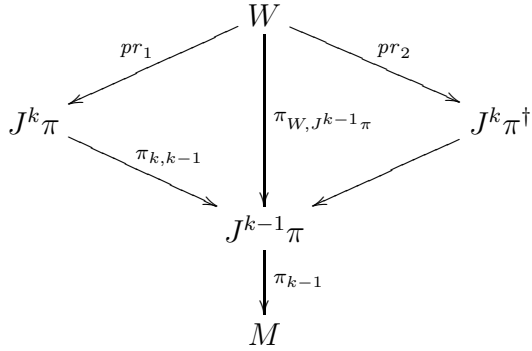

Assume that $\LL:J^k\pi\To\Lambda^mM$ is a Lagrangian density. Together with the pairing $\Phi$ (Equation \eqref{eq:natural.pairing}), we use this Lagrangian $\LL$ to define the dynamical map $\HH$ on $W$ (corresponding to the Hamiltonian density):
\begin{equation} \label{eq:hamiltonian}
\HH = \Phi - \LL\circ\pr_1.
\end{equation}

Consider the canonical multisymplectic $(m+1)$-form $\Omega$ on $J^k\pi^\dag$ (Equation \eqref{eq:canonical.form}), whose pullback to $W$ shall be denoted also by $\Omega$. We define on $W$ the premultisymplectic $(m+1)$-form
\begin{equation} \label{eq:hamiltonian.form}
\Omega_\HH = \Omega + \diff\HH.
\end{equation}
In adapted coordinates
\begin{eqnarray}
\HH &=& \prth{p^{Ii}_\alpha u^\alpha_{I+1_i}+p-L(x^i,u^\alpha_I,u^\alpha_K)}\dmx\\ \label{eq:hamiltonian.form.coord}
\Omega_\HH &=& -\diff p^{Ii}_\alpha\wedge\du^\alpha_I\wedge\dmxi+\prth{p^{Ii}_\alpha\du^\alpha_{I+1_i}+u^\alpha_{I+1_i}\diff p^{Ii}_\alpha-\pp[L]{u^\alpha_J}\du^\alpha_J}\wedge\dmx,
\end{eqnarray}
where $|I|=0,\dots,k-1$ and $|J|=0,\dots,k$.

We search for a  $\pi_{W,M}$-transverse and locally decomposable $m$-multivector field $X$ on $W$ (see Appendix \ref{sec:multivectors}) that is solution of the \emph{dynamical equation}:
\begin{equation} \label{eq:dynamical.equation}
i_X\Omega_\HH=0.
\end{equation}

\begin{remark} \label{rmk:dynamical.equation.projector}
The approach that we followed in \cite{CmpsLnMrt09} was to solve the equation
\begin{equation} \label{eq:dynamical.equation.projector}
i_\hh\Omega_\HH = (m-1)\Omega_\HH,
\end{equation}
where $\hh:TW\To TW$ is the horizontal projector associated to a horizontal distribution in $\pi_{W,M}:W\To M$ (\ie\ a connection). In fact, equations \eqref{eq:dynamical.equation} and \eqref{eq:dynamical.equation.projector} are equivalent. Two solutions $X$ and $\hh$ of the corresponding equation are related in such a way that $X$ generates the horizontal distribution to which $\hh$ is associated. If a solution $\hh$ of \eqref{eq:dynamical.equation.projector} is locally written in the following form (compare with the expression \eqref{eq:multivector} below)
\begin{equation} \label{eq:projector}
\hh=\dx^j\tensor\prth{\pp{x^j}+A^\alpha_{Jj}\pp{u^\alpha_J}+B^{Ii}_{\alpha j}\pp{p^{Ii}_{\alpha j}}+C_j\pp{p}},
\end{equation}
then its components will be governed by the relations \eqref{eq:holonomy}, \eqref{eq:dynamics.bottom}, \eqref{eq:dynamics.mid}, \eqref{eq:dynamics.top}, \eqref{eq:tangency.w1} and \eqref{eq:tangency.h0}.
\end{remark}

\begin{remark} \label{rmk:dynamical.equation.h}
Another formulation of the dynamical equation \eqref{eq:dynamical.equation} is
\begin{equation} \label{eq:dynamical.equation.h}
(-1)^mi_X\Omega = -\diff H,
\end{equation}
which is the regular expression for classical mechanics (\ie\ whenever $M$ has dimension $m=1$). The only difference is that, in this latter one, the tangency condition to $\HH=0$ is already included: $\Omega$ posses explicitly $\diff p$ in its local expression, while $\Omega_\HH$ does not. So \eqref{eq:dynamical.equation} may not determine any condition on the coefficients $C_j$ of $X_j$, which is equivalent to establish whether $X$ is tangent to $\HH=0$ or not. 
\end{remark}

It is shown that solutions $X$ of \eqref{eq:dynamical.equation} do not exist on the whole $W$, \cf\ \cite{CmpsLnMrt09,LnMrrMrt03,MrnMnzRmn03}. Because of that, we need to restrict the equation to the space where solutions do exist, that is to
\begin{equation} \label{eq:w1.def}
W_1 = \set{ w\in W\ :\ \exists X\in\Lambda^m_dT_wW\ \textrm{such that}\ \Lambda^mT_w\pi_{W,M}(X)\neq0\ \textrm{and}\ i_X\Omega_\HH(w)=0 }.
\end{equation}
Thus, a solution $X$ of \eqref{eq:dynamical.equation} rather than being a multivector field on $W$, will be a multivector field along $W_1$.

\begin{remark} \label{rmk:volume.form}
It is important to point out that the sole purpose of the base volume form $\eta$ is to ensure the existence of non-vanishing $m$-multivector fields on $M$, which implies the existence of $\pi_{W,M}$-transverse $m$-multivector fields on $W$.
\end{remark}

An arbitrary $\pi_{W,M}$-transverse multivector $X\in\Lambda^m_dT_wW$ may be written in the form
\begin{equation} \label{eq:multivector}
X = f\Lambda_{j=1}^mX_j = f\Lambda_{j=1}^m\prth{\pp{x^j}+A^\alpha_{Jj}\pp{u^\alpha_J}+B^{Ii}_{\alpha j}\pp{p^{Ii}_{\alpha j}}+C_j\pp{p}}.
\end{equation}
If we furthermore impose the condition $\eta(X)=1$ ($f=1$) and we compute the local expression of the equation \eqref{eq:dynamical.equation}, we then obtain
\begin{eqnarray}
\label{eq:holonomy}
A^\alpha_{Ii} &=& u^\alpha_{I+1_i}, \textrm{ with } |I|=0,\dots,k-1,\ i=1,\dots,\\
\label{eq:dynamics.bottom}
0  &=& \pp[L]{u^\alpha} - B^{\ j}_{\alpha j}; \\
\label{eq:dynamics.mid}
\sum_{I+1_i=J}p^{Ii}_\alpha &=& \pp[L]{u^\alpha_J} - B^{Jj}_{\alpha j}, \textrm{ with } |J|=1,\dots,k-1;\\
\label{eq:dynamics.top}
\sum_{I+1_i=K}p^{Ii}_\alpha &=& \pp[L]{u^\alpha_K}, \textrm{ with } |K|=k.
\end{eqnarray}

We shall refer to the equation \eqref{eq:holonomy} as the \emph{equation of holonomy} and to the equations \eqref{eq:dynamics.bottom}, \eqref{eq:dynamics.mid} and \eqref{eq:dynamics.top} as the \emph{equations of dynamics}. We will further named the equations of dynamics after the order of the multi-index in them, being the \emph{bottom}, \emph{mid} and \emph{top} level ones, respectively. We could have written these in the more compact way
\begin{equation} \label{eq:dynamics.concise}
\sum_{I+1_i=J}p^{Ii}_\alpha = \pp[L]{u^\alpha_J} - B^{Jj}_{\alpha j}, \textrm{ with } |J|=0,\dots,k,
\end{equation}
understanding that the first summation term is empty when $|J|=0$, as well as it is the last one when $|J|=k$ (there are no $B^{Ki}_{\alpha j}$ with $|K|=k$).

Notice that the top level equation of dynamics \eqref{eq:dynamics.top} is a constraint on the point $w\in W$ where the multivector $X$ stands, thus it defines the submanifold $W_1$ of $W$. The remaining equations, equation \eqref{eq:holonomy}, \eqref{eq:dynamics.bottom} and \eqref{eq:dynamics.mid}, are restrictions on the coefficients of the multivector $X$ . Note also that, for the time being, the $A$'s with greatest order multi-index and the $C$'s remain undetermined, as well as the most part of the $B$'s.

The existence of a $\pi_{W,M}$-transverse and locally decomposable $m$-multivector $X\in\Lambda^m_dT_wW$ (with $\eta(X)=1$) is guaranteed for every point $w\in W_1$. However, we cannot assure that such multivector $X$ is ``tangent'' to $W_1$ at $w$. Here, ``tangent'' means that the distribution associated to $X$ is tangent to $W_1$ at $w$. Therefore, we require $X$ to be tangent to $W_1$ by imposing the tangency condition
\[ X_j\prth{\sum_{I+1_i=K}p^{Ii}_\alpha - \pp[L]{u^\alpha_K}} = 0. \]
Furthermore, we also require $X$ to be tangent to the submanifold of $W$
\begin{equation} \label{eq:w0}
W_0 = \set{ w\in W\,:\,\HH(w)=0 } = \set{ w\in W\,:\,p = L-p^{Ii}_\alpha u^\alpha_{I+1_i}}
\end{equation}
by imposing the second tangency condition
\[ X_j\prth{p+p^{Ii}_\alpha u^\alpha_{I+1_i}-L} = 0. \]
These two previous conditions are equivalent to the following \emph{equations of tangency}
\begin{align}
\label{eq:tangency.w1}
\sum_{I+1_i=K}B^{Ii}_{\alpha j} =\,& \frac{\partial^2L}{\partial x^j\partial u^\alpha_K} + \sum_{|I|=0}^{k-1}u^\beta_{I+1_j}\frac{\partial^2L}{\partial u^\beta_I\partial u^\alpha_K} + \sum_{|R|=k}A^\beta_{Rj}\frac{\partial^2L}{\partial u^\beta_R\partial u^\alpha_K}, \textrm{ with } |K|=k;\\
\label{eq:tangency.h0}
C_j =\,& \pp[L]{x^j} + \sum_{|J|=0}^{k-1}u^\alpha_{J+1_j}\rbra\pp[L]{u^\alpha_I}-\!\!\sum_{I+1_i=J}p^{Ii}_\alpha\rket - \sum_{|I|=0}^{k-1}u^\alpha_{I+1_i}B^{Ii}_{\alpha j}.
\end{align}
Therefore, the coefficients of the multivector $X$ are governed by the equations of holonomy \eqref{eq:holonomy}, the bottom and mid level equations of dynamics \eqref{eq:dynamics.bottom} and \eqref{eq:dynamics.mid}, and the equations of tangency \eqref{eq:tangency.w1} and \eqref{eq:tangency.h0}.

Looking closer to the first equation of tangency \eqref{eq:tangency.w1}, we may observe that, if the matrix of second order partial derivatives of $L$ with respect to the ``velocities'' of highest order
\begin{equation}\label{eq:regular.lagrangian}
\prth{\frac{\partial^2L}{\partial u^\beta_R\partial u^\alpha_K}}_{|R|=|K|=k}
\end{equation}
is non-degenerate, then the highest order $A$'s are completely determined in terms of the highest order $B$'s. In the sequel, we will say that the Lagrangian $\LL:J^k\pi\To\Lambda^mM$ is \emph{regular} if, for any system of adapted coordinates, the matrix \eqref{eq:regular.lagrangian} is non-degenerate.

Note also that, thanks to the Lemma \ref{th:lower.sum} and the top level equation of dynamics \eqref{eq:dynamics.top}, the terms in the equation of tangency \eqref{eq:tangency.h0} that have $A$'s with a multi-index of length $k$ cancel out, while the terms that have $A$'s with a lower multi-index are already determined (equation of holonomy \eqref{eq:holonomy}). So, in some sense, the $C$'s depend only on the $B$'s.

\begin{theorem}
Consider the submanifold $W_2=W_0\cap W_1\stackrel{i}{\hookrightarrow}W$ and let $\Omega_2=i^*\Omega_\HH$. We have that $(W_2,\Omega_2)$ is multisymplectic if and only if $\LL$ is regular.
\end{theorem}

\begin{theorem} \label{th:euler-lagrange.geometric}
Let $\sigma\in\Sections\pi_{W,M}$ be an integral section of a solution $X$ of the dynamical equation \eqref{eq:dynamical.equation}. Then, its ``Lagrangian part'' $\sigma_1=\pr_1\!\circ\,\sigma$ is holonomic, $\sigma_1=j^k\phi$ for some section $\phi\in\Sections\pi$, which furthermore satisfies the higher-order Euler-Lagrange equations \eqref{eq:euler-lagrange}.
\end{theorem}

\begin{theorem}
Consider the system of linear equations in $B's$ (of highest order) given by the mid level equation of dynamics \eqref{eq:dynamics.mid} and the tangency condition \eqref{eq:tangency.w1}. This system has always maximal rank but, it is overdetermined when $k=1$ or $m=1$, completely determined when $k=m=2$, and undetermined otherwise.
\end{theorem}

\section{Variational Calculus  with Constraints} \label{sec:variational.calculus.constraints}
We consider a constraint submanifold $i:\cns\hookrightarrow J^k\pi$ of codimension $l$, which is locally annihilated by $l$ functionally independent constraint functions $\Psi^\mu$, where $1\leq\mu\leq l$. The constraint submanifold $\cns$ is supposed to fiber over the whole of $M$.

\begin{remark} \label{rmk:constraint.reduction}
As our ultimate goal is to find holonomic jet sections that belong to $\cns$, one could look for a submanifold $\cns'$ of $\cns$ consisting of the image of such sections. The submanifold $\cns'$ is given by the constraint functions of $\cns$ plus their consequences up to order $k$, that is, $\Psi^\mu$, $\dd[\Psi^\mu]{x^i}$, $\dd[^2\Psi^\mu]{x^{ij}}$, etc. Geometrically, $\cns'$ is obtained as the output of the following recursive process:
\begin{equation} \label{eq:constraint.reduction}
\cns^{(s,r)} := \cbra\begin{array}{rll}
                                          \cns,& s=0,& r=k;\\
                       \pi_{k,0}(\cns^{s-1,k}),& s>0,& r=0;\\
J^1\cns^{(s,r-1)}\cap\pi_{k,r}(\cns^{(s-1,k)}),& s>0,& 0<r\leq k;
\end{array}\nket
\end{equation}
which stops when, for some step $s\geq0$, $\cns^{(s+1,k)}=\cns^{(s,k)}$. This algorithm is a generalization to jet bundles of the method given in \cite{MrmMndTlcz95} by Mendella \etal\ to extract the integral part of a differential equation in a tangent bundle. The reader is also refereed to the alternative approach by Gasqui \cite{Gsq85}.

For instance, if one considers the null divergence restriction $u_x+v_y=0$ in the 2nd-order jet manifold of $\pr_1:\RR^3\times\RR^2\To\RR^3$, then the resulting manifold $\cns^{(2,2)}=\cns^{(1,2)}$ is given by the restrictions $u_x+v_y=0$, $u_{xt}+v_{yt}=0$, $u_{xx}+v_{xy}=0$ and $u_{xy}+v_{yy}=0$ (see Example \ref{ex:navier-stokes}).
\end{remark}

We now look for extremals of the Lagrangian action \eqref{eq:integral.action} restricted to those sections $\phi\in\Sections\pi$ whose $k$-jet takes values in $\cns$ (see \cite{FrnGrcRdr04,MrsdPkrsSkll01}). We will use the Lagrange multiplier theorem that follows.

\begin{theorem}[Abraham, Marsden \& Ratiu \cite{AbrhMrsdRat88}] \label{th:lagrange.multipliers}
Let $\MM$ be a smooth manifold, $f:\MM\To\RR$ be $\CC^r$, $r\geq1$, $\FF$ a Banach space, $g:\MM\To\FF$ a smooth submersion and $\NN=g^{-1}(0)$. A point $\phi\in N$ is a critical point of $f|_N$ if and only if there exists $\lambda\in\FF^*$, called a \emph{Lagrange multiplier}, such that $\phi$ is a critical point of $f-\abra\lambda,g\aket$.
\end{theorem}

In order to apply the Lagrange multiplier theorem, we need to define constraints as the 0-level set of some function $g$. We configure therefore the following setting: choose the smooth manifold $\MM$ to be the space of local sections $\Sections_R\pi=\set{\phi:R\subset M\to E:\pi\circ\phi=\Id_M}$, for some compact region $R\subset M$. The Banach space $\FF$ is the set of smooth functions $\smooth(R,\RR^l)$, provided with the $L^2$-norm, for instance. The constraint function $\Psi$ induces a constraint function on the space of local sections $\Sections_R\pi$ by mapping each section $\phi$ to the evaluation of its $k$-lift by the constraint, that is,
\[ g:\phi\in\Sections_R\pi\mapsto\Psi\circ j^k\phi\in\smooth(R,\RR^l). \]
Note that the 0-level set $\NN=g^{-1}(0)$ is the set of sections whose $k$-lift takes values in the constraint manifold $\cns$ (over $R$).

We therefore obtain that a section $\phi:M\To E$ is a critical point of the integral action $\Lact$ restricted to $\cns$ if and only if there exists a Lagrange multiplier $\lambda\in(\smooth(R,\RR^l))^*$ such that $\phi$ is a critical point of $\Lact-\abra\lambda,g\aket$. A priori, we cannot assure that the pairing $\abra\lambda,g(\phi)\aket$ has an integral expression of the type $\int_R\lambda_\mu\Psi^\mu\circ j^k\phi\dmx$ for some functions $\lambda_\mu:R\To\RR$. Henceforth, we shall suppose that that is the case.

\begin{remark} \label{rmk:abnormality}
In Theorem \ref{th:lagrange.multipliers} appears some regularity conditions that exclude the so-called \emph{abnormal solutions}. In general, given a critical point $\phi\in {\mathcal N}=g^{-1}(0)$  of $f_{|{\mathcal N}}$ , the classical Lagrange multiplier theorem claims that there exists a nonzero element  $(\lambda_0, \lambda)\in\RR\times\FF^*$ such that $\phi$ is a critical point of
\begin{equation}\label{eq:abnormality}
\lambda_0f - \abra\lambda,g\aket.
\end{equation}
Under the submersivity condition on  $g$, that is $\phi$ is a \emph{regular critical point},  it is possible to guarantee that $\lambda_0\not=0$ and dividing by $\lambda_0$ in \eqref{eq:abnormality} we obtain the characterization of critical points given in Theorem \ref{th:lagrange.multipliers}. The critical points $\phi$ with vanishing Lagrange multiplier, that is, $\lambda_0=0$ are called \emph{abnormal critical points}.

In the sequel we will only study the regular critical points, but our developments are easily adapted for the case of abnormality (adding the Lagrange multiplier $\lambda_0$ and studying separately both cases, $\lambda_0=0$ and $\lambda_0=1$).
\end{remark}

\begin{proposition}[Constrained higher-order Euler-Lagrange equations] \label{th:cns.euler-lagrange}
Let $\phi\in\Sections\pi$ be a critical point of the Lagrangian action $\Lact$ given in \eqref{eq:integral.action} restricted to those sections of $\pi$ whose $k$th lift take values in the constraint submanifold $\cns\subset J^k\pi$. If the associated Lagrange multiplier $\lambda$ is regular enough, then there must exist $l$ smooth functions $\lambda_\mu:R\subset M\To\RR$ that satisfy together with $\phi$ the \emph{constrained higher-order Euler-Lagrange equations}
\begin{equation} \label{eq:cns.euler-lagrange}
(j^{2k}\phi)^*\prth{ \sum_{|J|=0}^k(-1)^{|J|}\dd[^{|J|}]{x^J}\prth{\pp[L]{u^\alpha_J}-\lambda_\mu\pp[\Psi^\mu]{u^\alpha_J}} } = 0.
\end{equation}
\end{proposition}

\section{Constrained Mechanics in Higher Order Field Theories} \label{sec:constrained.mechanics}
As in the previous section, we begin by considering a constraint submanifold $i:\cns\hookrightarrow J^k\pi$ of codimension $l$, which is locally annihilated by $l$ functionally independent constraint functions $\Psi^\mu$, where $1\leq\mu\leq l$. The constraint submanifold $\cns$ is supposed to fiber over the whole of $M$ and it is not necessarily generated from a previous constraint submanifold by the process shown in Remark \ref{rmk:constraint.reduction}. We define in the restricted velocity-momentum space $W_0=\set{w\in W:\HH(w)=0}$ the constrained velocity-momentum space $W_0^\cns=\pr_1^{-1}(\cns)$, which is a submanifold of $W_0$, whose induced embedding and whose constraint functions will still be denoted $i:W_0^\cns\hookrightarrow W$ and $\Psi^\mu$, where $1\leq\mu\leq l$. The first order case $k=1$ is treated in \cite{CntrVnkr07}.

The following proposition allows us to work in local coordinates on the unconstrained velocity-momentum space $W$.

\begin{proposition}
Given a point $w\in W_0^\cns$, let $X\in\Lambda^m_d(T_wW_0^\cns)$ be a decomposable multivector and denote its image, $i_*(X)\in\Lambda^m(T_wW)$, by $\bar X$. The following statements are equivalent:
\begin{enumerate}
\item $i_X\Omega_0^\cns(Y)=0$ for every $Y\in T_wW_0^\cns$;
\item $i_{\bar X}\Omega\in T^0_wW_0^\cns$;
\end{enumerate}
where $T^0_wW_0^\cns$ is the annihilator of $i_*(T_wW_0^\cns)$ in $T_wW$.
\end{proposition}

We therefore look for solutions of the \emph{constrained dynamical equation}
\begin{equation} \label{eq:dynamical.equation.aux}
(-1)^mi_{\bar X}\Omega = -\lambda_\mu\diff\Psi^\mu-\lambda\diff H,
\end{equation}
where $\bar X$ is a tangent multivector field along $W^\cns_0$, the $\lambda^\mu$'s and $\lambda$ are Lagrange multipliers to be determined. Here, the coefficient $(-1)^m$ is used for technical purposes.

\begin{remark} \label{rmk:lagrange.multipliers}
It should be said that the Lagrange multipliers that appear in the dynamical equation  \eqref{eq:dynamical.equation.aux} have a different nature that the ones that appear in Proposition \ref{th:cns.euler-lagrange}. The former are locally defined on $W$, while the latter are locally defined on $M$. Although they coincide on the integral sections $\sigma\in\Sections\pi_{W,M}$ of a solution $X$ of the dynamical equation \eqref{eq:dynamical.equation.aux}, since its ``Lagrangian part'' $\sigma_1=\pr_1\circ\,\sigma$ satisfies the constrained Euler-Lagrange equation \eqref{eq:cns.euler-lagrange} with $\tilde\lambda_\mu=\lambda_\mu\circ\sigma$ (\cf\ Proposition \ref{th:cns.euler-lagrange.geometric}).
\end{remark}

Let $\bar X\in\Lambda^m(T_wW)$ be a decomposable $m$-vector at a given point $w\in W$, that is,  $\bar X=\bar X_1\wedge\dots\wedge\bar X_m$ for $m$ tangent vectors $\bar X_i\in T_wW$, which have the form
\begin{equation} \label{eq:tangent.vector}
\bar X_j = \pp{x^j}+A^\alpha_{Jj}\pp{u^\alpha_J}+B^{Ii}_{\alpha j}\pp{p^{Ii}_\alpha}+C_j\pp{p}
\end{equation}
in a given adapted chart $(x^i,u^\alpha_J,p^{Ii}_\alpha,p)$. A straightforward computation gives us
\begin{equation}
(-1)^mi_{\bar X}(\diff p^{Ii}_\alpha\wedge\du^\alpha_I\wedge\dmxi) = (A^\alpha_{Ii}B^{Ij}_{\alpha j}-A^\alpha_{Ij}B^{Ij}_{\alpha i})\dx^i+A^\alpha_{Ii}\diff p^{Ii}_\alpha-B^{Ii}_{\alpha i}\du^\alpha_I
\end{equation}
and
\begin{equation}
(-1)^mi_{\bar X}(\diff p\wedge\dmx) = \diff p - C_i\dx^i.
\end{equation}
Applying the above equations to the dynamical one \eqref{eq:dynamical.equation.aux}, we obtain the relations
\[\begin{array}{lrcl}
\textrm{coefficients in }\diff p:
& 1 &=& \lambda;\\
\textrm{coefficients in }\diff p^{Ii}_\alpha:
& A^\alpha_{Ii} &=& \lambda u^\alpha_{I+1_i};\\
\textrm{coefficients in }\du^\alpha_J
& B^{\ i}_{\alpha i} &=& \lambda\pp[L]{u^\alpha}-\lambda_\mu\pp[\Psi^\mu]{u^\alpha};\\
&  B^{Ii}_{\alpha i} &=& \lambda\prth{\pp[L]{u^\alpha_I}-\sum_{J+1_j=I}p^{Jj}_\alpha}-\lambda_\mu\pp[\Psi^\mu]{u^\alpha_I};\\
&                  0 &=& \lambda\prth{\pp[L]{u^\alpha_K}-\sum_{J+1_j=K}p^{Jj}_\alpha}-\lambda_\mu\pp[\Psi^\mu]{u^\alpha_K};\\
\textrm{coefficients in }\dx^j:
& A^\alpha_{Ii}B^{Ii}_{\alpha j}-A^\alpha_{Ij}B^{Ii}_{\alpha i}+C_j &=& \lambda\pp[L]{x^j}-\lambda_\mu\pp[\Psi^\mu]{x^j}.
\end{array}\]
Thus, a decomposable $m$-vector $\bar X\in\Lambda^m(T_wW)$ at a point $w\in W$ is a solution of the dynamical equation
\begin{equation} \label{eq:cns.dynamical.equation}
(-1)^mi_{\bar X}\Omega = -\lambda_\mu\diff\Psi^\mu-\diff H,
\end{equation}
if for any adapted chart $(x^i,u^\alpha_J,p^{Ii}_\alpha,p)$, the coefficients of $\bar X$ and the point $w$ satisfy the equations
\begin{align}
\label{eq:cns.holonomy.lag}
A^\alpha_{Ii} =\,& u^\alpha_{I+1_i}, \textrm{ with } |I|=0,\dots,k-1,\ i=1,\dots;\\
\label{eq:cns.dynamics.lag.bottom}
                          0 =\,& \pp[L]{u^\alpha}-\lambda_\mu\pp[\Psi^\mu]{u^\alpha}-B^{\ j}_{\alpha j};\\
\label{eq:cns.dynamics.lag.mid}
\sum_{I+1_i=J}p^{Ii}_\alpha =\,& \pp[L]{u^\alpha_J}-\lambda_\mu\pp[\Psi^\mu]{u^\alpha_J}-B^{Jj}_{\alpha j}, \textrm{ with } |J|=1,\dots,k-1;\\
\label{eq:cns.dynamics.lag.top}
\sum_{I+1_i=K}p^{Ii}_\alpha =\,& \pp[L]{u^\alpha_K}-\lambda_\mu\pp[\Psi^\mu]{u^\alpha_K}, \textrm{ with } |K|=k;\\
\label{eq:cns.tangency.lag.h}
C_j =\,& \pp[L]{x^j}-\lambda_\mu\pp[\Psi^\mu]{x^j} + \sum_{|J|=0}^{k-1}u^\alpha_{J+1_j}\rbra\pp[L]{u^\alpha_I}-\lambda_\mu\pp[\Psi^\mu]{u^\alpha_J}-\!\!\sum_{I+1_i=J}p^{Ii}_\alpha\rket - u^\alpha_{I+1_i}B^{Ii}_{\alpha j}.
\end{align}

Because of the Lagrange multipliers $\lambda_\mu$, we cannot describe the submanifold of $W^\cns_0$ where solutions $X$ of the constrained dynamical equation \eqref{eq:cns.dynamical.equation} exist, like it has been done in \eqref{eq:dynamics.top} for the unconstrained dynamical equation \eqref{eq:dynamical.equation}. Therefore, we need to get rid off of them. Consider the more concise expression for the equations of dynamics \eqref{eq:cns.dynamics.lag.bottom}, \eqref{eq:cns.dynamics.lag.mid} and \eqref{eq:cns.dynamics.lag.top}
\begin{equation}
\label{eq:cns.dynamics.lag.concise}
\sum_{I+1_i=J}p^{Ii}_\alpha = \pp[L]{u^\alpha_J}-\lambda_\mu\pp[\Psi^\mu]{u^\alpha_J}-B^{Jj}_{\alpha j}, \textrm{ with } |J|=0,\dots,k,
\end{equation}
where, as in \eqref{eq:dynamics.concise}, the first summation term is understood to be void when $|J|=0$, as well as it is the last one when $|J|=k$. We now suppose that the constraints $\Psi^\mu$ are of the type $u^{\hat\alpha}_{\hat J}=\Phi^{\hat\alpha}_{\hat J}(x^i,u^{\check\alpha}_{\check J})$, where $u^{\hat\alpha}_{\hat J}$ are some constrained coordinates which depend on the free coordinates $(x^i,u^{\check\alpha}_{\check J})$ through the functions $\Phi^{\hat\alpha}_{\hat J}$. Thus, the constraint have the form $\Psi^{\hat\alpha}_{\hat J}(x^i,u^\alpha_J)=u^{\hat\alpha}_{\hat J}-\Phi^{\hat\alpha}_{\hat J}(x^i,u^{\check\alpha}_{\check J})=0$. So, writing again the previous equation \eqref{eq:cns.dynamics.lag.concise} for the different sets of coordinates, the ones that are free and the ones that are not, we obtain
\begin{eqnarray}
\label{eq:cns.dynamics.lag.concise.constrained}
\sum_{I+1_i=\hat J}p^{Ii}_{\hat\alpha} &=& \pp[L]{u^{\hat\alpha}_{\hat J}}-\lambda^{\hat J}_{\hat\alpha}\phantom{\pp[\Phi_J]{u_J}}-B^{\hat Jj}_{\hat\alpha j}, \textrm{ with } |\hat J|=0,\dots,k;\\
\label{eq:cns.dynamics.lag.concise.free}
\sum_{I+1_i=\check J}p^{Ii}_{\check\alpha} &=& \pp[L]{u^{\check\alpha}_{\check J}}-\lambda^{\hat J}_{\hat\alpha}\pp[\Phi^{\hat\alpha}_{\hat J}]{u^{\check\alpha}_{\check J}}-B^{\check Jj}_{\check\alpha j}, \textrm{ with } |\check J|=0,\dots,k.
\end{eqnarray}
Substituting $-\lambda^{\hat J}_{\hat\alpha}$ from \eqref{eq:cns.dynamics.lag.concise.constrained} into \eqref{eq:cns.dynamics.lag.concise.free}, we have that
\begin{equation} \label{eq:cns.dynamics.concise}
\sum_{I+1_i=\check J}p^{Ii}_{\check\alpha}+\prth{\sum_{I+1_i=\hat J}p^{Ii}_{\hat\alpha}}\pp[\Phi^{\hat\alpha}_{\hat J}]{u^{\check\alpha}_{\check J}} = \pp[L]{u^{\check\alpha}_{\check J}}+\pp[L]{u^{\hat\alpha}_{\hat J}}\pp[\Phi^{\hat\alpha}_{\hat J}]{u^{\check\alpha}_{\check J}}-B^{\check Jj}_{\check\alpha j}-B^{\hat Jj}_{\hat\alpha j}\pp[\Phi^{\hat\alpha}_{\hat J}]{u^{\check\alpha}_{\check J}}, \textrm{ with } |\check J|=0,\dots,k.
\end{equation}
Note that, when $|\check J|=k$, the term $B^{\check Jj}_{\check\alpha j}$ disappears, but $B^{\hat Jj}_{\hat\alpha j}\pp[\Phi^{\hat\alpha}_{\hat J}]{u^{\check\alpha}_{\check J}}$ do not necessarily. This is circumvent by supposing that, if $|\hat J|<k$, then $\pp[\Phi^{\hat\alpha}_{\hat J}]{u^\alpha_K}=0$ for any $|K|=k$. That is the case when the constraint submanifold $\cns$ has no constraint of higher order, \ie\ $\cns=\pi^{-1}_{k,k-1}(\pi_{k,k-1}(\cns))$, or, more generally, when $\cns$ fibers by $\pi_{k,k-1}$ over its image.

Taking this into account, we expand the previous equation \eqref{eq:cns.dynamics.concise}, obtaining then constrained equations of dynamics freed of the Lagrange multipliers
\begin{eqnarray}
\label{eq:cns.dynamics.bottom}
\sum_{I+1_i=\hat J}p^{Ii}_{\hat\alpha}\pp[\Phi^{\hat\alpha}_{\hat J}]{u^{\check\alpha}} &\!\!=\!\!& \pp[L^\cns]{u^{\check\alpha}}-B^{\ j}_{\check\alpha j}-B^{\hat Jj}_{\hat\alpha j}\pp[\Phi^{\hat\alpha}_{\hat J}]{u^{\check\alpha}};\\
\label{eq:cns.dynamics.mid}
\sum_{I+1_i=\check J}p^{Ii}_{\check\alpha}+\!\!\!\sum_{I+1_i=\hat J}p^{Ii}_{\hat\alpha}\pp[\Phi^{\hat\alpha}_{\hat J}]{u^{\check\alpha}_{\check J}} &\!\!=\!\!& \pp[L^\cns]{u^{\check\alpha}_{\check J}}-B^{\check Jj}_{\check\alpha j}-B^{\hat Jj}_{\hat\alpha j}\pp[\Phi^{\hat\alpha}_{\hat J}]{u^{\check\alpha}_{\check J}}, \textrm{ with } |\check J|=1,\dots,k-1;\\
\label{eq:cns.dynamics.top}
\sum_{I+1_i=\check K}\!p^{Ii}_{\check\alpha}+\!\!\!\sum_{I+1_i=\hat K}\!p^{Ii}_{\hat\alpha}\pp[\Phi^{\hat\alpha}_{\hat K}]{u^{\check\alpha}_{\check K}} &\!\!=\!\!& \pp[L^\cns]{u^{\check\alpha}_{\check K}}, \textrm{ with } |\check K|=k;
\end{eqnarray}
where $\pp[L^\cns]{u^{\check\alpha}_{\check J}}=\pp[L]{u^{\check\alpha}_{\check J}}+\pp[L]{u^{\hat\alpha}_{\hat J}}\pp[\Phi^{\hat\alpha}_{\hat J}]{u^{\check\alpha}_{\check J}}$, being $\LL^\cns=\LL\circ i:\cns\To\Lambda^mM$ the restricted Lagrangian.

We are now in disposition to define the submanifold $W^\cns_2$ along to which solutions of the constrained dynamical equation \eqref{eq:cns.dynamical.equation} exist,
\begin{equation} \label{eq:wc2}
W^\cns_2 = \set{ w\in W^\cns_0\ :\ \eqref{eq:cns.dynamics.top} }
         = \set{ w\in W\ :\
\begin{array}{c}
u^{\hat\alpha}_{\hat J}=\Phi^{\hat\alpha}_{\hat J}(x^i,u^{\check\alpha}_{\check J})\\
\displaystyle \phantom{\pp{u_J}} p=L(x^i,u^\alpha_J)-p^{Ii}_\alpha u^\alpha_{I+1_i} \phantom{\pp{u_J}}\\
\displaystyle
\sum_{I+1_i=\check K}\!p^{Ii}_{\check\alpha}+\!\!\!\sum_{I+1_i=\hat K}\!p^{Ii}_{\hat\alpha}\pp[\Phi^{\hat\alpha}_{\hat K}]{u^{\check\alpha}_{\check K}} = \pp[L^\cns]{u^{\check\alpha}_{\check K}}
\end{array} }
\end{equation}

Tangency conditions on $X$ with respect to $W^\cns_2$ will give us the constrained equations of tangency
\begin{align}
\label{eq:cns.tangency.cns}
A^{\hat\alpha}_{\hat Jj} =\,& \pp[\Phi^{\hat\alpha}_{\hat J}]{x^j} + A^{\check\alpha}_{\check Jj}\pp[\Phi^{\hat\alpha}_{\hat J}]{u^{\check\alpha}_{\check J}},\\
\label{eq:cns.tangency.h0}
C_j
=\,& \pp[L^\cns]{x^j} + \sum_{|\check J|=0}^{k-1}u^{\check\alpha}_{\check J+1_j}\rbra\pp[L^\cns]{u^{\check\alpha}_{\check J}}-\!\!\sum_{I+1_i=\check J}p^{Ii}_{\check\alpha}\rket\\
\nonumber
 \,& - \sum_{|\hat J|=0}^{k-1}u^{\hat\alpha}_{\hat J+1_j}\!\!\sum_{I+1_i=\hat J}p^{Ii}_{\hat\alpha} - \sum_{|\hat K|=k}\pp[\Phi^{\hat\alpha}_{\hat K}]{x^j}\!\!\sum_{I+1_i=\hat K}p^{Ii}_{\hat\alpha} - \sum_{|I|=0}^{k-1}u^\alpha_{I+1_i}B^{Ii}_{\alpha j},\\
\label{eq:cns.tangency.w1}
\sum_{I+1_i=\check K}\!\!B^{Ii}_{\check\alpha j}
=\,&
\frac{\partial^2L^\cns}{\partial x^j \partial u^{\check\alpha}_{\check K}} -
\sum_{I+1_i=\hat K}\!p^{Ii}_{\hat\alpha}\frac{\partial^2\Phi^{\hat\alpha}_{\hat K}}{\partial x^j \partial u^{\check\alpha}_{\check K}}\\
\nonumber
 \,& + A^{\check\beta}_{\check Jj} \prth{
\frac{\partial^2L^\cns}{\partial u^{\check\beta}_{\check J} \partial u^{\check\alpha}_{\check K}} - \!\!\!
\sum_{I+1_i=\hat K}\!p^{Ii}_{\hat\alpha}\frac{\partial^2\Phi^{\hat\alpha}_{\hat K}}{\partial u^{\check\beta}_{\check J} \partial u^{\check\alpha}_{\check K}} } - \!\!\!
\sum_{I+1_i=\hat K}\!\!B^{Ii}_{\hat\alpha j}\pp[\Phi^{\hat\alpha}_{\hat K}]{u^{\check\alpha}_{\check K}},\ |\check K|=k.
\end{align}

\begin{proposition} \label{th:multisymplectic.iff.cns-regular}
Let $\Omega^\cns_2$ be the pullback of the premultisymplectic form $\Omega_\HH$ to $W^\cns_2$ by the natural inclusion $i:W^\cns_2\hookrightarrow W$, that is $\Omega^\cns_2=i^*(\Omega_\HH)$. Suppose that $m=\dim M>1$, then the $(m+1)$-form $\Omega^\cns_2$ is multisymplectic if and only if $\LL$ is \emph{regular along $W^\cns_2$}, \ie\ if and only if the matrix
\begin{equation} \label{eq:cns-hessian}
\prth{\frac{\partial^2L^\cns}{\partial u^{\check\beta}_{\check R} \partial u^{\check\alpha}_{\check K}} - \!\!\!
\sum_{I+1_i=\hat K}\!p^{Ii}_{\hat\alpha}\frac{\partial^2\Phi^{\hat\alpha}_{\hat K}}{\partial u^{\check\beta}_{\check R} \partial u^{\check\alpha}_{\check K}}}_{|\check R|=|\check K|=k}
\end{equation}
is non-degenerate along $W^\cns_2$.
\end{proposition}

\begin{proof}
First of all, let us make some considerations. By definition, $\Omega^\cns_2$ is multisymplectic whenever $\Omega^\cns_2$ has trivial kernel, that is,
\[ \textrm{if}\ v\in TW_2,\ i_v\Omega^\cns_2=0\ \Longleftrightarrow\ v=0\ . \]
This is equivalent to say that
\[ \textrm{if}\ v\in i_*(TW_2),\ i_v\Omega_\HH|_{i_*(TW_2)}=0\ \Longleftrightarrow\ v=0\ . \]

Let $v\in TW$ be a tangent vector whose coefficients in an adapted basis are given by
\[ v = \gamma^i\frac{\partial}{\partial x^i} + A^\alpha_J\frac{\partial}{\partial u^\alpha_J} + B^{Ii}_\alpha\frac{\partial}{\partial p^{Ii}_\alpha} + C\frac{\partial}{\partial p}. \]
Using the expression \eqref{eq:canonical.form.coord}, we may compute the contraction of $\Omega_\HH$ by $v$,
\begin{equation} \label{eq:ivomega}
\begin{split}
i_v\Omega_\HH =
& - B^{Ii}_\alpha\du^\alpha_I\wedge\dmxi + A^\alpha_I\diff p^{Ii}_\alpha\wedge\dmxi- \gamma^j\diff p^{Ii}_\alpha\wedge\du^\alpha_I\wedge\diff^{m-2}x_{ij}\\
& + \prth{A^\alpha_{I+1_i}p^{Ii}_\alpha + B^{Ii}_\alpha u^\alpha_{I+1_i} - A^\alpha_J\pp[L]{u^\alpha_J}}\dmx\\
& - \gamma^j\prth{p^{Ii}_\alpha\du^\alpha_{I+1_i} + u^\alpha_{I+1_i}\diff p^{Ii}_\alpha - \pp[L]{u^\alpha_J}\du^\alpha_J}\wedge\diff^{m-1}x_j.
\end{split}
\end{equation}

In addition to this, let us consider a vector $v\in TW$ tangent to $W_2$, that is $v\in i_*(TW_2)$, we then have that
\[ \diff(u^{\hat\alpha}_{\hat J}-\Phi^{\hat\alpha}_{\hat J})(v)=0,
\quad
\diff\prth{\sum_{I+1_i=\check K}\!p^{Ii}_{\check\alpha}+\!\!\!\sum_{I+1_i=\hat K}\!p^{Ii}_{\hat\alpha}\pp[\Phi^{\hat\alpha}_{\hat K}]{u^{\check\alpha}_{\check K}} - \pp[L^\cns]{u^{\check\alpha}_{\check K}}}(v)=0
\quad\textrm{and}\quad
\diff H(v)=0, \]
which leads us to the following relations for the coefficients of $v$:
\begin{eqnarray}
\label{eq:dpsiv}
A^{\hat\alpha}_{\hat J} &=&
\gamma^j\pp[\Phi^{\hat\alpha}_{\hat J}]{x^j}
+ A^{\check\alpha}_{\check J}\pp[\Phi^{\hat\alpha}_{\hat J}]{u^{\check\alpha}_{\check J}},\\
\label{eq:dw1v}
\sum_{I+1_i=\check K}\!\!B^{Ii}_{\check\alpha} &=&
\gamma^j \prth{
\frac{\partial^2L^\cns}{\partial x^j \partial u^{\check\alpha}_{\check K}}
- \!\!\! \sum_{I+1_i=\hat K}\!p^{Ii}_{\hat\alpha}\frac{\partial^2\Phi^{\hat\alpha}_{\hat K}}{\partial x^j \partial u^{\check\alpha}_{\check K}} }\\
\nonumber
& &
+ A^{\check\beta}_{\check J} \prth{
\frac{\partial^2L^\cns}{\partial u^{\check\beta}_{\check J} \partial u^{\check\alpha}_{\check K}} - \!\!\!
\sum_{I+1_i=\hat K}\!p^{Ii}_{\hat\alpha}\frac{\partial^2\Phi^{\hat\alpha}_{\hat K}}{\partial u^{\check\beta}_{\check J} \partial u^{\check\alpha}_{\check K}} }
- \!\!\! \sum_{I+1_i=\hat K}\!\!B^{Ii}_{\hat\alpha}\pp[\Phi^{\hat\alpha}_{\hat K}]{u^{\check\alpha}_{\check K}}\\
\label{eq:dhv}
C &=&
\gamma^j \prth{
\pp[L^\cns]{x^j}
- \!\!\! \sum_{I+1_i=\hat K}\!p^{Ii}_{\hat\alpha}\pp[\Phi^{\hat\alpha}_{\hat K}]{x^j} }\\
\nonumber
& &
+ A^{\check\alpha}_{\check J} \prth{
\pp[L^\cns]{u^{\check\alpha}_{\check J}}
- \!\!\! \sum_{I+1_i=\check J}\!p^{Ii}_{\check\alpha}
- \!\!\! \sum_{I+1_i=\hat J}\!p^{Ii}_{\hat\alpha}\pp[\Phi^{\hat\alpha}_{\hat J}]{u^{\check\alpha}_{\check J}} }
- B^{Ii}_\alpha u^\alpha_{I+1_i}.
\end{eqnarray}
It is important to note that, even though in all the previous equations \eqref{eq:ivomega}, \eqref{eq:dpsiv}, \eqref{eq:dw1v} and \eqref{eq:dhv} explicitly appear $A$'s with multi-index of length $k$, for such a vector $v\in i_*(TW_2)$, the terms associated to these $A$'s cancel out in the development of $i_v\Omega_\HH$, Equation \eqref{eq:ivomega}, and the third tangency relation \eqref{eq:dhv}. Thus, a tangent vector $v\in i_*(TW_2)$ would kill $\Omega_\HH$ if and only if its coefficients satisfy the following relations
\[\begin{array}{c}
\displaystyle
\gamma^j=0,\ A^\alpha_I=0,\ B^{Ii}_\alpha=0,\ C=0,\\
\displaystyle
A^{\hat\alpha}_{\hat K}=A^{\check\alpha}_{\check K}\pp[\Phi^{\hat\alpha}_{\hat K}]{u^{\check\alpha}_{\check K}}\ \textrm{and}\ 
A^{\check\beta}_{\check R} \prth{
\frac{\partial^2L^\cns}{\partial u^{\check\beta}_{\check R} \partial u^{\check\alpha}_{\check K}} - \!\!\!
\sum_{I+1_i=\hat K}\!p^{Ii}_{\hat\alpha}\frac{\partial^2\Phi^{\hat\alpha}_{\hat K}}{\partial u^{\check\beta}_{\check R} \partial u^{\check\alpha}_{\check K}} } = 0.
\end{array}\]

These considerations being made, the assertion is now clear.
\end{proof}

\begin{proposition} \label{th:cns.euler-lagrange.geometric}.
Let $\sigma\in\Sections\pi_{W,M}$ be an integral section of a solution $X$ of the constrained dynamical equation \eqref{eq:cns.dynamical.equation}. Then, its ``Lagrangian part'' $\sigma_1=\pr_1\!\circ\,\sigma$ is holonomic, $\sigma_1=j^k\phi$ for some section $\phi\in\Sections\pi$, which furthermore satisfies the constrained higher-order Euler-Lagrange equations \eqref{eq:cns.euler-lagrange}.
\end{proposition}

\begin{proof}
If $X$ is locally expressed as in \eqref{eq:multivector}, we know that it must satisfy the equations of dynamics \eqref{eq:cns.dynamics.lag.bottom},  \eqref{eq:cns.dynamics.lag.top} and  \eqref{eq:cns.dynamics.lag.top}, for unknown Lagrange multipliers $\lambda_\mu$. If we note $\lambda_\mu'=\lambda_\mu\circ\sigma$ and $L'=L-\lambda_\mu'\Psi^\mu$, it suffices to follow the demonstration for $L'$ of Theorem \ref{th:euler-lagrange.geometric} which is proven in \cite{CmpsLnMrt09}.
\end{proof}

\section{Example}

Here, we study an incompressible fluid under control as in \cite{AgrSry08}. The corresponding equations are the Navier-Stokes one plus the divergence-free condition:
\begin{eqnarray}
\label{eq:navier-stokes}
\pp[\vv]t + \nabla_\vv\vv + \nabla\Pi &=& \nu\Delta\vv + \ff\\
\label{eq:divergence-free}
                       \nabla\cdot\vv &=& 0
\end{eqnarray}
where the vector field $\vv$ is the velocity of the fluid, $\ff$ is the field of exterior forces acting on the fluid, which will be our controls, and the scalar functions $\Pi$ and $\nu$ are the pressure and the viscosity, respectively. In particular, our case of interest is the two dimensional case on $\RR^2$ endowed with the standard metric. If we fix global Cartesian coordinates $(x,y)$ on $\RR^2$ and adapted coordinates $(x,y,u,v)$ on its tangent $T\RR^2=\RR^4$, the previous equations become
\begin{eqnarray}
\label{eq:navier-stokes.coord.a}
u_t + u\cdot u_x + v\cdot u_y + \partial_x\Pi &=& \nu\cdot(u_{xx}+u_{yy}) + F\\
\label{eq:navier-stokes.coord.b}
v_t + u\cdot v_x + v\cdot v_y + \partial_y\Pi &=& \nu\cdot(v_{xx}+v_{yy}) + G\\
\label{eq:divergence-free.coord}
                                    u_x + v_y &=& 0
\end{eqnarray}
where, with some abuse of notation, $\vv(t,x,y)=(u,v)$ and $\ff=(F,G)$.

We therefore look for time-dependent vector fields $\vv=(u,v)$ on $\RR^2$ that satisfy the Navier-Stokes equations \eqref{eq:navier-stokes.coord.a} and \eqref{eq:navier-stokes.coord.b} for a prescribed control $\ff=(F,G)$ and submitted to the free divergence condition \eqref{eq:divergence-free.coord}. Moreover, we look for such vector fields $\vv=(u,v)$ that are in addition optimal in the controls for the integral action
\begin{equation} \label{eq:integral.action.ns}
\Lact(\vv,R) = \frac12\int_R\norm\ff^2\dt\wedge\dx\wedge\dy.
\end{equation}

In order to apply the development of the present jet bundle framework, all of this is restated in the following way: We set a fiber bundle $\pi:E\To M$ by putting $M=\RR\times\RR^2$, $E=\RR\times T\RR^2$ and $\pi=(\pr_1,\pr_{\RR^2})$. We fix global adapted coordinates $(t,x,y,u,v)$ on $E$, which induce the corresponding global adapted coordinates on $J^k\pi$ and $J^k\pi^\dag$. Besides, we choose the volume form $\eta$ on $M$ to be $\dt\wedge\dx\wedge\dy$. Thus, the Lagrangian function $L:J^2\pi\to\RR$ is nothing else but
\[ L = \frac12(F^2+G^2), \]
where we obtain $F$ and $G$ as functions on $J^2\pi$ using the equations \eqref{eq:navier-stokes.coord.a} and \eqref{eq:navier-stokes.coord.b}.

To make the reading easier, we change slightly the coordinate notation of jet bundles to fit in this example: The coordinate ``velocities'' associated to $u$ and $v$ will still be labeled $u$ and $v$, respectively, with symmetric subindexes (as in the original equations); the coordinate ``momenta'' associated to $u$ and $v$ will now be labeled $p$ and $q$, respectively, with non-symmetric subindexes. Finally and as we will focus on the equations
of dynamics \eqref{eq:cns.dynamics.bottom}, \eqref{eq:cns.dynamics.mid} and \eqref{eq:cns.dynamics.top}, the coefficients in the local expression \eqref{eq:multivector} of a multivector $X$ associated to the coordinate momenta $p$ and $q$ will be labeled $B$ and $D$, respectively.

\begin{example}[The Euler equation] \label{ex:euler.equation}
We will first suppose that the fluid is Eulerian, that is, it has null viscosity. In this case, the Lagrangian function $L=(F^2+G^2)/2$ associated to the integral action \eqref{eq:integral.action.ns} is of first order when the ``Euler equations'', \eqref{eq:navier-stokes.coord.a} and  \eqref{eq:navier-stokes.coord.b} with $\nu=0$, are taken into account. In $J^1\pi$, we consider the divergence-free constraint submanifold $\cns=\set{z\in J^1\pi\,:\,u_x+u_y=0}$, which introduces a single Lagrange multiplier $\lambda$.

Proceeding with the theoretical machinery, we compute the bottom level equations of dynamics corresponding to those of \eqref{eq:dynamics.bottom}
\begin{eqnarray*}
0 &=& u_x\cdot F + v_x\cdot G - (B^t_t+B^x_x+B^y_y)\\
0 &=& u_y\cdot F + v_y\cdot G - (D^t_t+D^x_x+D^y_y)
\end{eqnarray*}
and the top level equations of dynamics (there are no middle ones) corresponding to those of \eqref{eq:dynamics.top}
\begin{align*}
p^t =\,& F                  & q^t =\,& G\\
p^x =\,& u\cdot F - \lambda & q^x =\,& u\cdot G\\
p^y =\,& v\cdot F           & q^y =\,& v\cdot G - \lambda
\end{align*}
We can dispose of the only Lagrange multiplier $\lambda$ by putting
\[ p^x-q^y = u\cdot F-v\cdot G, \]
what defines $W^\CC_1$ together with the top level equations of dynamics with no Lagrange multiplier.

From here, we may compute also the constrained Euler-Lagrange equations \eqref{eq:cns.euler-lagrange} for this problem, which are
\begin{align*}
\diff_tF + u\cdot\diff_xF + v\cdot\diff_yF + v_y\cdot F - v_x\cdot G &= \partial_x\lambda\\
\diff_tG + u\cdot\diff_xG + v\cdot\diff_yG + u_x\cdot G - u_y\cdot F &= \partial_y\lambda
\end{align*}
where $\diff_*=\dd*$.

Finally, we note that $L$ is not regular along $W^\cns_2$ since the square matrix, that correspond to \eqref{eq:cns-hessian},
\[\rbra\begin{array}{ccccc}
 1 & u & v & 0 & 0 \\
 u & u^2+v^2 & u\cdot v & -v & -u\cdot v \\
 v & u\cdot v & v^2 & 0 & 0 \\
 0 & -v & 0 & 1 & u \\
 0 & -u\cdot v & 0 & u & u^2
\end{array}\rket\]
has obviously rank 2. Here we have used as $u_x$ as independent (``check'') coordinate and $v_y$ as dependent (``hat'') coordinate.
\end{example}

\begin{example}[The Navier-Stokes equation] \label{ex:navier-stokes}
Now, we tackle the full problem of the Navier-Stokes equations. In this case, the Lagrangian function $L=(F^2+G^2)/2$ is of second order. In $J^2\pi$, we consider the constraint submanifold
\[ \cns = \set{z\in J^2\pi\ : \ u_x+u_y=0,\ u_{tx}+v_{ty}=0,\ u_{xx}+v_{xy}=0,\ u_{xy}+v_{yy}=0 } \]
which comes from the first order constraint \eqref{eq:divergence-free}, free divergence, and its consequences to second order (see Remark \ref{rmk:constraint.reduction}). These constraints introduce for Lagrange multiplier $\lambda$, $\lambda_t$, $\lambda_x$ and $\lambda_y$ that are associated to them respectively.

We now proceed like in the previous example by computing the equations of dynamics. In first place, we have the bottom level ones corresponding to those of \eqref{eq:dynamics.bottom}
\begin{eqnarray*}
0 &=& u_x\cdot F + v_x\cdot G - (B^t_t+B^x_x+B^y_y)\\
0 &=& u_y\cdot F + v_y\cdot G - (D^t_t+D^x_x+D^y_y)
\end{eqnarray*}
Note that they are formally the same as before. In second place, the mid level equations corresponding to those of \eqref{eq:dynamics.mid}
\begin{align*}
p^t =\,&        F - (B^{tt}_t+B^{tx}_x+B^{ty}_y)&
q^t =\,&        G - (D^{tt}_t+D^{tx}_x+D^{ty}_y)\\
p^x =\,& u\cdot F - (B^{xt}_t+B^{xx}_x+B^{xy}_y) + \lambda&
q^x =\,& u\cdot G - (D^{xt}_t+D^{xx}_x+D^{xy}_y)\\
p^y =\,& v\cdot F - (B^{yt}_t+B^{yx}_x+B^{yy}_y)&
q^y =\,& v\cdot G - (D^{yt}_t+D^{yx}_x+D^{yy}_y) + \lambda
\end{align*}
Note that formally they also coincide with the top level ones of the previous example but for the coefficients that now appear in them. And in third place, the top level equations corresponding to those of \eqref{eq:dynamics.top}
\begin{align*}
       p^{tt} =\,& 0                       &        q^{tt} =\,& 0\\
       p^{xx} =\,& -\nu\cdot F - \lambda_x &        q^{xx} =\,& -\nu\cdot G\\
       p^{yy} =\,& -\nu\cdot F             &        q^{yy} =\,& -\nu\cdot G - \lambda_y\\
p^{tx}+p^{xt} =\,& -\lambda_t              & q^{tx}+q^{xt} =\,& 0\\
p^{ty}+p^{yt} =\,& 0                       & q^{ty}+q^{yt} =\,& -\lambda_t\\
p^{xy}+p^{yx} =\,& -\lambda_y              & q^{xy}+q^{yx} =\,& -\lambda_x
\end{align*}
We can again get rid easily of the Lagrange multipliers by putting
\[p^{tx}+p^{xt}=q^{ty}+q^{yt} \qquad p^{xx}+\nu\cdot F=q^{xy}+q^{yx} \qquad p^{xy}+p^{yx}=q^{yy}+\nu\cdot G \]
what defines $W^\CC_1$ together with the top level equations of dynamics with no Lagrange multiplier.

From here, we may compute also the constrained Euler-Lagrange equations \eqref{eq:cns.euler-lagrange} for this problem, which are
\begin{eqnarray*}
2\partial^2_{tx}\lambda_t + \partial^2_{xx}\lambda_x + 2\partial^2_{xy}\lambda_y -
\partial_x\lambda
&=&  \partial^2_{xx}\nu\cdot F + 2\partial_x\nu\cdot\diff_xF + \nu\cdot\diff^2_{xx}F +\\
& &+ \partial^2_{yy}\nu\cdot F + 2\partial_y\nu\cdot\diff_yF + \nu\cdot\diff^2_{yy}F -\\
& &- \diff_tF - u\cdot\diff_xF - v\cdot\diff_yF - v_y\cdot F + v_x\cdot G\\
2\partial^2_{ty}\lambda_t + 2\partial^2_{xy}\lambda_x + \partial^2_{yy}\lambda_y -
\partial_y\lambda
&=&  \partial^2_{xx}\nu\cdot G + 2\partial_x\nu\cdot\diff_xG + \nu\cdot\diff^2_{xx}G +\\
& &+ \partial^2_{yy}\nu\cdot G + 2\partial_y\nu\cdot\diff_yG + \nu\cdot\diff^2_{yy}G -\\
& &- \diff_tG - u\cdot\diff_xG - v\cdot\diff_yG - u_x\cdot G + u_y\cdot F
\end{eqnarray*}
As before, the Lagrangian is not regular along $W^\cns_2$, what seems to be clear if we observe that $L$ is highly non-degenerate: It depends only on 4 of the 12 coordinates of second order. It is worthless to show its ``Hessian'', even though it is interesting to say that it is null only when $\nu$ is.
\end{example}

\section{Conclusions and future work}

In this paper, we have introduced an unambiguous  geometric formalism for higher order field theories subjected to constraints, with applications to optimal control of partial differential equations. Our theory is based on the classical Skinner and Rusk formalism and the theory of higher-order jet bundles. In the future, we will do a detailed study of how to derive a constraint algorithm derive, in particular we will pay attention on the necessary conditions in order to the algorithm do not stop at the primary constraint submanifold; for instance, restricting the dual jet bundle to an appropriate submanifold (in agreement with the ideas behind \cite{SndCrmp90}). The design of variational integrators for these higher order field equations will be also analyzed (see \cite{MrsdPkrsSkll01} for the first order case). 

Besides, it is worth mentioning now that the BRST construction for field theory has been used systematically to deal with the problem of quantizing Gauge Systems (\cf\ \cite{HnnxTtlb92}). Gauge systems are an instance of constrained field theories where the constraints arise because of the Gauge invariance of the system. The BRST construction is based on the introduction of auxiliary odd and even degrees of freedom (referred as ``ghosts'' in physics literature) in such a way that the supermanifold structure of the extended formalism allows for a cohomological description of the reduced spaces, i.e., of the relevant degrees of freedom from a physical point of view (see for instance \cite{IbrtMrn94} for a simple description of the reduction procedure in terms of supergeometry). This construction can be put in a nice geometrical setting for first order Lagrangian mechanics (see for instance \cite{FrgKlnd92}) and we are looking forward to extend such formalism for higher order Lagrangian mechanics and Lagrangian field theories.

\appendix

\section{Multivectors} \label{sec:multivectors}

Let $P$ be a $n$-dimensional differentiable manifold. Sections of $\Lambda^m(TP)$ (with $1\leq m\leq n$) are called \emph{$m$-multivector fields} in $P$. We will denote by $\Fields^m(P)$ the set of $m$-multivector fields in $P$. Given $X\in\Fields^m(P)$, for every $p\in P$, there exists an open neighborhood $U_p\subset P$ and $X_1,\ldots ,X_r\in\Fields(U_p)$ such that
\[ X\stackrel{U_p}{=}\sum_{1\leq i_1<\ldots <i_m\leq r} f^{i_1\ldots i_m}X_{i_1}\wedge\ldots\wedge X_{i_m} \]
with $f^{i_1\ldots i_m}\in\smooth(U_p)$ and $m\leq r\leq n$. A multivector field $X\in\Fields^m(P)$ is \emph{locally decomposable} if, for every $p\in P$, there exists an open neighborhood $U_p\subset P$ and $X_1,\ldots ,X_m\in\Fields (U_p)$ such that
\[ X\stackrel{U_p}{=}X_1\wedge\ldots\wedge X_m. \]
We will denote by $\Fields^m_d(P)$ the set of locally decomposable $m$-multivector fields in $P$.

Let $D\subseteq TP$ be an $m$-dimensional distribution. The sections of $\Lambda^mD$ are locally decomposable $m$-multivector fields in $P$. A locally decomposable $m$-multivector field $X\in\Fields^m_d(P)$ and an $m$-dimensional distribution $D\subseteq TP$ are \emph{associated} whenever $X$ is a section of $\Lambda^mD$. If $X,X'\in\Fields^m_d(P)$ are non-vanishing multivector fields associated with the same distribution $D$, then there exists a non-vanishing function $f\in\smooth(P)$ such that $X'=fX$. This fact defines an equivalence relation in the set of non-vanishing $m$-multivector fields in $P$, whose equivalence classes will be denoted by $\DD(X)$. There is a bijective correspondence between the set of $m$-dimensional orientable distributions $D$ in $TP$ and the set of the equivalence classes $\DD(X)$ of non-vanishing, locally decomposable $m$-multivector fields $X$ in $P$. By abuse of notation, $\DD(X)$ will also denote the $m$-dimensional orientable distribution $D$ in $TP$ with whom $X$ is associated.

An $m$-dimensional submanifold $S\hookrightarrow P$ is said to be an \emph{integral manifold} of $X\in\Fields^m(P)$ (resp. of an $m$-dimensional distribution $D$ in $TP$) if $X$ spans $\Lambda^mTS$ (resp. if $TS=D$). In such a case, $X$ (resp. $D$) is said to be \emph{integrable}. Integrable multivector fields shall be locally decomposable. A non-vanishing, locally decomposable $m$-multivector $X\in\Fields^m_d(P)$ is \emph{involutive} if its associated distribution $\DD(X)$ is involutive, that is, if $\sbra\DD(X),\DD(X)\sket\subseteq\DD(X)$. If a non-vanishing multivector field $X\in\Fields^m_d(P)$ is involutive, so is every other in its equivalence class $\DD(X)$. By Frobenius' theorem, a non-vanishing and locally decomposable multivector field is integrable if, and only if, it is involutive.

Now, let $\pi:P\To M$ be a fiber bundle with $\dim M=m$. A multivector field $X\in\Fields^m(P)$ is said to be \emph{$\pi$-transverse} if $\Lambda^m\pi_*(X)$ does not vanish at any point of $M$, hence $M$ must be orientable. If $X\in\Fields^m(P)$ is integrable, then $X$ is $\pi$-transverse if, and only if, its integral manifolds are sections of $\pi:P\To M$. In this a case, if $S$ is an integral manifold of $X$, then there exists a section $\phi\in\Sections\pi$ such that $S=\mathrm{Im}(\phi)$.

For more details on multivector fields and their relation with field theories, refer to \cite{MnzRmn98,MnzRmn99}.

\section{The multi-index notation} \label{sec:multi-indexes}

Given a function $f:\RR^m\To\RR$, its partial derivatives are classically denoted
\[ f_{i_1i_2\cdots i_k} = \frac{\partial^kf}{\partial x_{i_1}\partial x_{i_2}\cdots\partial x_{i_k}}. \]
When smooth functions are considered, their crossed derivatives coincide. Thus, the order in which the derivatives are taken is not important, but the number of times with respect to each variable.

Another notation to denote partial derivatives is defined through ``symmetric'' multi-indexes (see \cite{Snd89}). A multi-index $I$ will be an $m$-tuple of non-negative integers. The $i$-th component of $I$ is denoted $I(i)$. Addition and subtraction of multi-indexes are defined component-wise (whenever the result is still a multi-index), $(I\pm J)(i)=I(i)\pm J(i)$. The length of $I$ is the sum $|I|=\sum_iI(i)$, and its factorial $I!=\Pi_iI(i)!$. In particular, $1_i$ will be the multi-index that is zero everywhere except at the $i$-th component which is equal to 1.

Keeping in mind the above definition, we shall denote the partial derivatives of a function $f:\RR^m\To\RR$ by:
\[ f_I = \pp[^{|I|}f]{x^I} = \frac{\partial^{I(1)+I(2)+\dots+I(m)}f}{\partial x_1^{I(1)}\partial x_2^{I(2)}\cdots\partial x_m^{I(m)}}. \]
Thus, given a multi-index $I$, $I(i)$ denotes the number of times the function is differentiated with respect to the $i$-th component. The former notation should not be confused with the latter one. For instance, the third order partial derivative $\frac{\partial^3f}{\partial x_2\partial x_3\partial x_2}$ (with $f:\RR^4\To\RR$) is denoted $f_{232}$ and $f_{(0,2,1,0)}$, respectively.

Here we present some simple but useful results on multi-indexes.

\begin{lemma} \label{th:fubini}
Let $\set{a_{I,i}}_{I,i}$ be a family of real numbers indexed by a multi-index $I\in\nat^m$ and by an integer $i$ such that $1\leq i\leq m$. Given an integer $k\geq1$, we have that
\begin{equation} \label{eq:fubini}
\sum_{|I|=k-1} \sum_{i=1}^m a_{I,i} = \sum_{|J|=k} \sum_{I+1_i=J} a_{I,i}.\end{equation}
\end{lemma}

\begin{lemma} \label{th:identity}
Let $J\in\nat^m$ be a non-zero multi-index. We have that
\begin{equation} \label{eq:identity}
\sum_{I+1_i=J} \frac{I(i)+1}{|I|+1} = 1.
\end{equation}
\end{lemma}

\begin{lemma} \label{th:lower.sum}
Let $\set{a_J}_J$ be a family of real numbers indexed by a multi-index $J\in\nat^m$. Given a positive integer $l\geq1$, we have that
\begin{equation} \label{eq:lower.sum}
\sum_{|J|=l} a_J = \sum_{|I|=l-1} \sum_{i=1}^m \frac{I(i)+1}{|I|+1}a_{I+1_i},\end{equation}
\end{lemma}

\begin{lemma} \label{th:lower.paired.sum}
Let $\set{a_J,b^J}_J$ be a family of real numbers indexed by a multi-index $J\in\nat^m$. Given an integer $l\geq1$, we have that
\begin{equation} \label{eq:lower.paired.sum}
\sum_{|J|=l} b^Ja_J = \sum_{|I|=l-1} \sum_{i=1}^m \frac{I(i)+1}{|I|+1}(b^{I+1_i}+Q^{I,i})a_{I+1_i},
\end{equation}
where $\set{Q^{I,i}}_{I,i}$ is a family of real numbers such that for any multi-index $J\in\nat^m$ (with $|J|\geq1$) we have that
\begin{equation} \label{eq:q-condition}
\sum_{I+1_i=J} \frac{I(i)+1}{|I|+1}Q^{I,i}=0.
\end{equation}
\end{lemma}

\section*{Acknowledgments}
This work has been partially supported by the MICINN, Ministerio de Ciencia e Innovación (Spain), project MTM2007-62478, project ``Ingenio Mathematica'' (i-MATH) No. CSD 2006-00032 (Consolider-Ingenio 2010). We want to thank specially professor Alberto Ibort Latre for bringing us some insight on the BRST formalism. The first author (C. C.) also acknowledges the MICINN for an FPI grant.


\end{document}